\documentclass{llncs}[11pt]
\usepackage{makeidx}  % allows for indexgeneration
\usepackage{graphicx}
\usepackage{graphics}

\renewcommand{\int}{\ensuremath {\mathcal I}}

\newcommand{\eqref}[1]{(\ref{#1})}

\newcommand{\fbk}{\texttt{ASSA}}
\newcommand{\stoller}{\texttt{{Stoller}}}

\begin{document}
\frontmatter          % for the preliminaries
\mainmatter              % start of the contributions
%
%%% OLD TITLE
% \title{A Symbolic Approach to the Parametric Analysis of ARBAC
%   Policies}
\title{{Automated Symbolic Analysis of ARBAC-Policies}}
\subtitle{(Extended Version)}

\titlerunning{Automated Analysis of ARBAC Policies} % abbreviated title (for running head)
% also used for the TOC unless
% \toctitle is used
%
\author{
Alessandro Armando\inst{1,2} \and Silvio Ranise\inst{2}
}
\authorrunning{Armando, Ranise}   % abbreviated author list (for running head)
\institute{%
  DIST, Universit\`a degli Studi di Genova, Italia \and
  Security and Trust Unit, FBK, Trento, Italia 
}
\maketitle              % typeset the title of the contribution
%
%
%		ABSTRACT
\begin{abstract}
  One of the most widespread framework for the management of
  access-control policies is Administrative Role Based Access Control
  (ARBAC).  Several automated analysis techniques have been proposed
  to help maintaining desirable security properties of ARBAC policies.
  One limitation of many available techniques is that the sets of
  users and roles are bounded.  In this paper, we propose a symbolic
  framework to overcome this difficulty.  We design an automated
  security analysis technique, parametric in the number of users and
  roles, by adapting recent methods for model checking infinite state
  systems that use first-order logic and state-of-the-art theorem
  proving techniques.  Preliminary experiments with a prototype
  implementations seem to confirm the scalability of our technique.
\end{abstract}
%
%
%		INTRODUCTION
\section{Introduction}
Role Based Access Control (RBAC)~\cite{rbac} regulates access by
assigning users to roles which, in turn, are granted permissions to
perform certain operations.  Administrative RBAC (ARBAC)~\cite{arbac}
specifies how RBAC policies may be changed by administrators; thus
providing support for decentralized policy administration, which is
crucial in large distributed systems.
% One of the most important problem in system security is the
% specification and management of access-control policies.  % Several
% frameworks have been put forward to specify complex policies.
% Because of their expressiveness, the problem of predicting if the
% policies correctly capture the intent of their designers is even
% more dramatic and automated analysis techniques are mandatory for
% systems of reasonable size to establish their correctness.  One of
% the most widespread framework for the specification and management
% of access-control policies---which are a crucial ingredient of
% system security---is ARBAC, Administrative Role Based Access control
% (see, e.g.,~\cite{arbac}).  In this framework, access is regulated
% by assigning users to roles which, in turn, are granted permissions
% to perform certain operations.  On top of this, administrative
% policies specify how the access policies may be changed by an
% administrator; thus providing support for decentralized policy
% administration, which is crucial in large distributed systems.  To
% enhance scalability, ARBAC has been extended in several directions
% by, e.g., adding delegation or permitting parameterized roles.
% These additional features make the need of automated analysis
% techniques of the policies even more acute.
For the sake of simplicity, we consider the URA97 component of
ARBAC97~\cite{arbac97}, which is concerned with the management of the
user-role assignment by administrative roles.  The generalization to
other variants of ARBAC is left to future work.

% Unfortunately, it is almost impossible for a human to foresee the
% subtle interplays between the operations carried out by different
% administrators because of the large number of possible interleavings.
% Even for real-world applications where only finitely many users and
% roles are considered, security analysis may be a daunting task because
% of their numbers.  
% For example, in~\cite{schaad01}, the ARBAC system of one of the major
% bank in Europe compries around forty thousounds users and over one
% thousand roles.
As it is almost impossible for a human to foresee the subtle
interplays between the operations carried out by different
administrators because of the large number of possible interleavings.
Automated analysis techniques are thus of paramount importance to
maintain desirable security properties while ensuring flexible
administration.  Several techniques have been proposed,
e.g.,~\cite{li-tripunitara,stoller,stoller2}.  % They are capable of
% solving \emph{goal reachability} problems of the form: given an
% initial policy state $S$, an administrative policy $P$, a set of
% administrators $A$, and a set of roles $R$, is it possible for the
% administrators in $A$ to modify the policy state $S$ (according to the
% policy $P$) so that a given user is a member of the roles in $R$?  
In general, security analysis problems are undecidable but become
decidable under suitable restrictions.  Indeed, the results of the
analysis are valid under the assumptions that make them decidable.  In
this respect, one of the most severe limitations of the available
techniques is that the number of users and roles is bounded, i.e.\
finite and known \emph{a priori}.  % For example, the work described
% in~\cite{stoller2} allows for an infinite set of roles but users must
% still be bounded.  
So, if one has proved that a certain property holds for, say, $1000$
users and $150$ roles and after some times, the number of users or
roles is changed for some reason, then the result of the previous
analysis no more holds and the automated technique must be invoked
again.  It would be desirable to have analysis techniques capable of
certifying that a certain property holds regardless of the number of
users or roles so to make their results more useful.

In this paper, we propose a symbolic framework to specify ARBAC
policies that enables the design of parametric (in the number of users
and roles) security analysis techniques.  The idea is to adapt recent
techniques for model checking infinite state systems~\cite{ijcar08}
that use decidable fragments of first-order logic and state-of-the-art
theorem proving techniques 
% (in particular, Satisfiability Modulo Theories (SMT) solvers)
to mechanize the analysis.  The paper makes two contributions towards
the goal of building parametric analysis techniques for ARBAC
policies.  The {former} is a \textbf{framework for the uniform
  specification of a variety of ARBAC policies}.  In particular, we
can describe security analysis problems where users and roles are
finitely many but their exact number is not known \emph{a priori}.
% In other words, we can specify goal reachability problems which are
% \emph{parametric} in the number of users or roles.  
The {second contribution} is a \textbf{symbolic} backward reachability
\textbf{procedure} that can be used \textbf{to solve} an important
class of security analysis problems, called \textbf{user-role
  reachability problems}, that allow one to check if certain users can
acquire a given permission or, dually, if a user can never be given a
role which would give him or her a permission which is not supposed to
have.  The security analysis problem is iteratively reduced to a
series of satisfiability checks in a decidable fragment of first-order
logic.  We use ideas from model theory and the theory of
well-quasi-ordering~\cite{ijcar08,AbdullaTCS} for the proof of
termination of the method, which turns out to be the most substantial
part of the proof of correctness.
% Both of these tests are reduced to solving (decidable)
% satisfiability problems for a class of first-order formulae used to
% represent sets of states.  The termination of the procedure is derived
% by exploiting some standard model-theoretic constructions that allow
% us to define a well-quasi-ordering on sets of states that implies
% termination of the backward reachability procedure (see,
% e.g.,~\cite{ijcar08,AbdullaTCS}).  
The decidability of the parametric goal reachability problem is
obtained as a corollary of the correctness of the procedure.  

% \paragraph{Related work.}
Our decidability result is more general that those
in~\cite{li-tripunitara,stoller} which assume a bounded number of
users and roles.  A comparison with the result in~\cite{stoller2} is
more articulated.  On the one hand, we are more general in allowing
for a finite but unknown number of users and roles while
in~\cite{stoller2} the users are bounded and only the roles are
parametric.  On the other hand, we allow for only a restricted form of
negation in the preconditions of certain administrative actions
while~\cite{stoller2} seems to allow for arbitrary negation.  We plan
to investigate how to extend our framework to allow for arbitrary
negation in the near future while in this paper we focus on the core
ideas.  Finally, our procedure can consider several initial RBAC
policies at the same time while~\cite{li-tripunitara,stoller,stoller2}
can handle only one.

\paragraph{Plan of the paper.}  In Section~\ref{sec:ARBAC-intro}, we
formally define ARBAC policies with their user-role reachability
problem.  In Section~\ref{sec:symbolic-specification-ARBAC}, we
present our symbolic framework for the specification of ARBAC polices.
In Section~\ref{sec:reachability-analysis}, we design a symbolic
analysis procedures of ARBAC policies.  In Section~\ref{sec:exp}, we
discuss some preliminary experiments with a prototype of our
technique.  Section~\ref{sec:conclusion} concludes and gives some
hints about future work.
%%% SR: COMMENT OR UNCOMMENT THE FOLLOWING LINES FOR THE SHORT OR LONG
%%% VERSION, RESP.
The omitted proofs and some additional material can be found in the
extended version of the paper~\cite{ext-version}.
%%% The omitted proofs and some additional material can be found in
%%% the Appendixes.

\section{RBAC and ARBAC policies}
\label{sec:ARBAC-intro}

We assume familiarity with ARBAC (see, e.g.,~\cite{arbac}) and
many-sorted first-order logic with equality (see,
e.g.,~\cite{enderton}).  Consider a signature
$\Sigma_{\mathit{ARBAC}}$ containing the sort symbols $\mathit{User},
\mathit{Role}$, and $\mathit{Permission}$, countably many constant
symbols $e^u_i, e^r_i, e^p_i$ (for $i\geq 0$) of sort $\mathit{User}$,
$\mathit{Role}$, and $\mathit{Permission}$, respectively, the
predicate symbols $\succeq$ (written infix), $pa$, and $ua$ of arity
$\mathit{Role}\times \mathit{Role}$, $\mathit{Role}\times
\mathit{Permission}$, and $\mathit{User}\times \mathit{Role}$,
respectively, and \emph{no function symbols}.  A \emph{RBAC policy} is
a first-order structure $\mathcal{M}=(D,I)$ over this signature, where
the interpretation of $ua$ (in symbols, $ua^I$) is the user-role
assignment relation, $pa^I$ is the permission-role assignment, and
$\succeq^I$ is the role hierarchy.  Without loss of generality, we
consider structures that interpret the sort symbols into (disjoint)
sets of users, roles, and permissions, respectively.
% (This is so because, although the interpretations of sorts (universes)
% are not required to be disjoint in many-sorted first-order logic,
% equality is defined only on elements of the same universe; hence,
% there will always exist an elementary equivalent structure---i.e.\
% satisfying the same set of formulae---whose universes are disjoint.)
Our notion of state corresponds to that of miniRBAC policy
in~\cite{stoller}.

An ARBAC policy prescribes how the user-role assignment, the
permission assignment, and the role hierarchy of RBAC policies may
evolve.  As in~\cite{stoller} and according to the URA97
administrative control model~\cite{arbac97}, in this paper, we assume
that the interpretations of $\succeq$ and $pa$ are constant over time
and only that of $ua$ may change.  We also assume that $\succeq^I$ is
a partial order and refer to $\succeq^I$ as the `more senior than'
relationship between roles.  % Syntactically, to emphasize that
% $ua$ is time variant, we write $ua(u,r)$ instead of $ua(u,r)$.
We abuse notation by denoting an interpretation $\mathcal{M}=(D,I)$
over $\Sigma_{\mathit{ARBAC}}$ with the restriction $s$ of $I$ to $ua$
when the rest of $\mathcal{M}$ is clear from the context and write
$\succeq$, $pa$, and $ua$ instead of $\succeq^I$, $pa^I$, and $ua^I$
(or $s(ua)$), respectively.

Let $s$ be a RBAC policy.  A user $u$ is an \emph{explicit member} of
a role $r$ in $s$ if $(u,r) \in s(ua)$ or, equivalently, $s \models
ua(u,r)$, where `$\models$' is the standard satisfaction relation of
many-sorted first-order logic.  Similarly, $u$ is an \emph{implicit
  member} of $r$ in $s$ if $(u,r')\in s(ua)$ for some $r'$ which is
more senior than $r$ or, equivalently, $s\models ua^*(u,r)$ where
$ua^*(u,r)$ abbreviates the formula $\exists r'.(r'\succeq r \wedge
ua(u,r'))$.  Thus, $u$ is \emph{not} a {member} of $r$ (neither
implicit nor explicit) if for all role $r'$ more senior than $r$, we
have $(u,r')\not\in s(ua)$ or, equivalently, $s\models \forall
r'.(r'\succeq r\Rightarrow \neg ua(u,r'))$.

% A user $u$ \emph{satisfies a URA97 constraint $\lambda_1\wedge
%   \cdots\wedge \lambda_n$ in the RBAC policy $s$} if, for each $i=1,
% ..., n$, $u$ is (not) an explicit member of role $r$ when $\lambda_i$
% is $ua(u,r)$ ($\neg ua(u,r)$, respectively) or if $u$ is (not) an
% implicit member of role $r$ when $\lambda_i$ is $ua^*[u,r)$ ($\neg
% ua^*(u,r)$, respectively) or, equivalently, $s\models \lambda_i$.  
A \emph{\sl can\_assign} action is a tuple $\langle r_a, C, r'
\rangle$ such that $r_a,r'$ are roles and $C$ is a (possibly empty)
finite set of \emph{role expressions} of the form $r$ or
$\overline{r}$ where $r$ is a role.  Sometimes, along the lines
of~\cite{li-tripunitara}, a set $T$ of users can be attached to a
\textsl{can\_assign} action; in this case, users in $T$ are assumed not to
initiate any role assignment.  A \emph{\sl can\_revoke} action is a
pair $\langle r_a, r' \rangle$ such that $r_a,r$ are roles.
%%%%%%%%%%%%%%%%%%%%%%%%%%%%%%%%%%%%%%%%%%%%%%%%%%%%%%%%%%%%%%%%%%%%%%%%
A \emph{user $u$ satisfies a role expression $\rho$ in a RBAC policy
  $s$} if $u$ is an implicit member of role $r$ in $s$ when $\rho$ is
$r$ (or, equivalently, $s\models ua^*(u,r)$) and $u$ is not a member
of role $r$ in $s$ when $\rho$ is $\overline{r}$ (or, equivalently,
$s\models \neg ua^*(u,r)$).  A \emph{user $u$ satisfies the finite set
  $C=\{\rho_1, ..., \rho_n\}$ of role expressions in a RBAC policy
  $s$} if $u$ satisfies $\rho_i$ in $s$, for each $i=1, ..., n$
($n\geq 0$) or, equivalently, $s\models [\neg] ua^*(u,r_1)\wedge
\cdots \wedge [\neg] ua^*(u,r_n)$, where $[\neg] ua^*(u,r_i)$ denotes
$ua^*(u,r_i)$ when $\rho_i$ is $r_i$ and $\neg ua^*(u,r_i)$ when
$\rho_i$ is $\overline{r_i}$.  If $n=0$, then $C=\emptyset$ and any
user $u$ always satisfies it.
%%%%%%%%%%%%%%%%%%%%%%%%%%%%%%%%%%%%%%%%%%%%%%%%%%%%%%%%%%%%%%%%%%%%%%%%
Let $s,s'$ be two RBAC policies.  A \emph{\textsl{can\_assign}} action $\langle
r_a, C, r' \rangle$ is \emph{enabled} in $s$ if there exist users
$u_a,u$ such that $u_a$ satisfies $r_a$ in $s$ and $u$ satisfies $C$
in $s$ and $s'$ is \emph{obtained} from $s$ by its application if
$s'(ua) = s(ua) \cup \{ (u,r')\}$.  A \emph{\textsl{can\_revoke}} action
$\langle r_a, r' \rangle$ is \emph{enabled} in $s$ if there exists a
user $u_a$ such that $u_a$ satisfies $r_a$ in $s$ and $s'$ is
\emph{obtained} from $s$ by its application if $s'(ua) = s(ua)
\setminus \{ (u,r')\}$.
% such that $r_a,r$ are roles and if $s,s'$ are two states,
% there exist users $u_a,u$ such that $(u_a,r_a)\in s(ua)$ (in which
% case, the \textsl{can\_revoke} is \emph{enabled}), then $s'(ua) = s(ua)
% \setminus \{ (u,r') \}$ (in which case, we say that $s'$ is
% \emph{obtained} by applying the action).  
If $\alpha$ is a \textsl{can\_assign} or a \textsl{can\_revoke} action, we write
$\alpha(s,s')$ to denote the fact that the action is enabled in $s$
and $s'$ is obtained from $s$ by applying $\alpha$.  The pair $(S_0,
A)$ is an \emph{ARBAC policy} when $S_0$ is a finite set of RBAC
policies, called \emph{initial}, and $A$ is a finite set of
\textsl{can\_assign} and \textsl{can\_revoke} actions.
% Let $\Gamma :=(S_0,\{\alpha_1, ..., \alpha_n\})$ be an ARBAC policy,
Let $u$ be a user, $RP$ be a finite set of pairs $(r,p)$ where $r$ is
a role and $p$ a permission.  The pair $\gamma:=(u,RP)$ is called the
\emph{goal} of the \emph{user-role reachability problem for
  $\Gamma:=(S_0, A)$} which consists of answering the following
question: is there a sequence $s_0, ..., s_m$ of states such that
$s_0\in S_0$, for each $i=0, ..., m-1$, there exists $\alpha\in A$ for
which $\alpha(s_i,s_{i+1})$, $(u,r)\in s_m(ua)$, and $(r,p)\in pa$ 
% (in symbols, $s_m\models ua(u,r)\wedge pa(r,p)$)
for each pair $(r,p)\in RP$.  If there is no such $m\geq 0$, then the
goal $\gamma$ is \emph{unreachable}; otherwise, it is \emph{reachable}
and the sequence $s_0, ..., s_m$ of states is called a \emph{run}
leading $\Gamma$ from an initial RBAC policy $s_0\in S_0$ to a RBAC
policy satisfying $\gamma$.

\begin{example}
  \label{ex:intro}
  We formalize the running example in~\cite{li-tripunitara}. Let
  $\mathcal{M}$ be an RBAC policy such that $\mathit{User} := \{
  Alice, Bob, Carol \}$, $\mathit{Permission} := \{ Edit, Access, View
  \}$, and $\mathit{Role} := \{ Employee,$ $Engineer, PartTime,
  FullTime, HumanResource,$ \linebreak $ProjectLead, and Manager
  \}$.\footnote{For the sake of clarity, here and the other examples
    of the paper, we will abuse notation by using more evocative names
    for constants than $e_i^{\alpha}$, $\alpha\in \{u,r,p\}$ ($i\geq
    0$).  Also, if constants have different identifiers, then they
    denote distinct elements. We use the same identifiers to denote
    constants and the elements they denote.}  Every user is a member
  of role Employee.  Managers work full-time. Project leaders are
  engineers.  Alice is an engineer who is part-time.  All employees
  have access permission to the office.  Thus, $\mathcal{M}$ is also
  such that $\succeq := \{(Engineer, Employee), (PartTime, Employee),
  (FullTime,$ \linebreak $Employee), (ProjectLead, Engineer),
  (Manager, FullTime)\}$, $pa := \{ (Access,$ $Employee), (View,
  HumanResource), (Edit, Engineer)\}$, $ua := \{ (Alice, Part\mbox{-}$
  \linebreak $Time), (Alice, Engineer), (Bob, Manager), (Carol,
  HumanResource)\}$.

  Examples of \textsl{can\_assign} are: $\langle Manager, \{
  Engineer,FullTime\}, ProjectLead \rangle$, $\langle HumanResource,
  \emptyset, FullTime \rangle$, and $\langle HumanResource, \emptyset,
  PartTime \rangle$.  The \linebreak meaning of the first action is
  that a manager can assign a full-time engineer to be a project
  leader; the second and the third ones mean that a user in the
  human-resources department can turn any user to be full-time or
  part-time.  If we attach to the previous assignments, the singleton
  set $T=\{Carol\}$ of users; then those actions cannot be performed
  by Carol even if she has the appropriate roles.  Examples of
  \textsl{can\_revoke} actions are: $\langle Manager, ProjectLead \rangle$,
  $\langle Manager,$ $Engineer \rangle$, $\langle HumanResource,$
  $FullTime \rangle$, and $\langle HumanResource,$ $PartTime \rangle$.
  For instance, the meaning of the first is that a manager can revoke
  the role of project leader to any user; the meaning of the other
  actions is similar.\qed
\end{example}

\section{Symbolic representation of ARBAC policies}
\label{sec:symbolic-specification-ARBAC}

Our framework represents (i) sets of RBAC policies as the models of a
first-order theory whose signature contains only constant and
predicate symbols but no function symbols, (ii) initial RBAC policies
and constraints as universal formulae, and goals of reachability
problems as existential formulae, and (iii) administrative actions
(such as the \textsl{can\_assign} and \textsl{can\_revoke}) as certain classes of
formulae.  The assumptions on the three components allow us to design
a decision procedure for the user-role reachability problem where the
number of users and roles is finite but unknown.  We now describe in
details these assumptions.

% In order to mechanize the solution of the user-role reachability
% problem, we introduce a symbolic representation of the structures
% representing sets of RBAC policies and of the \textsl{can\_assign} and `can
% revoke actions.  The idea is well-known in logic: first-order
% formulae can be used to describe classes of first-order structures.
% Indeed, to be useful, such formulae should allow for the automated
% solution of their satisfiability problem, since---as we will see
% below---the user-role reachability problem can be reduced to several
% of such problems.  Formally, we will use the notion of theory to
% identify classes of first-order structures.
% We assume a basic level of familiarity with many-sorted first-order
% logic~\cite{enderton}.  Each signature considered below includes the
% equality sign $=_S : S\times S$ (written infix), for each sort symbol
% $S$ in the signature, and it is to be interpreted as the identity
% relation.  By abuse of notation, we frequently omit the subscript $S$
% when it is clear from the context.  Let $\mathit{Users},
% \mathit{Roles},$ and $\mathit{Permissions}$ be sort symbols.
% We assume three non-empty, countably infinite, and pairwise disjoint
% sets of constants of sort $\mathit{Users}$, $\mathit{Roles}$, and
% $\mathit{Permissions}$.  
\paragraph{Formal preliminaries.}  A \emph{$\Sigma$-theory} is a set
of sentences (i.e.\ formulae where no free variables occur) over the
signature $\Sigma$.  A theory $T$ is axiomatized by a set $Ax$ of
sentences if every sentence $\varphi$ in $T$ is a logical consequence
of $Ax$.  We associate with $T$ the class $Mod(T)$ of structures over
$\Sigma$ which are models of the sentences in $T$.  A theory is
\emph{consistent} if $Mod(T)\neq \emptyset$. A $\Sigma$-formula
$\varphi$ is \emph{satisfiable modulo $T$} iff there exists
$\mathcal{M}\in Mod(T)$ such that $\mathcal{M}$ satisfies $\varphi$
(in symbols, $\mathcal{M}\models \varphi$).  A $\Sigma$-formula
$\varphi$ is \emph{valid modulo $T$} iff its negation is unsatisfiable
modulo $T$ and it is \emph{equivalent modulo $T$} to a
$\Sigma$-formula $\varphi'$ iff the formula $(\varphi \Leftrightarrow
\varphi')$ is valid modulo $T$.  As notational conventions,
% the constants of sort $\mathit{Users}$ are denoted by $e^u_i$, those
% of sort $\mathit{Roles}$ by $e^r_i$, and those of sort
% $\mathit{Permisions}$ by $e^p_i$ for $i\geq 0$.
the variables $u,r,p$ and their subscripted versions are of sort
$\mathit{Users}, \mathit{Roles}$, and $\mathit{Permissions}$,
respectively; $\underline{u},\underline{r}, \underline{p}$ denote
tuples of variables of sort $\mathit{Users},\mathit{Roles},
\mathit{Permission}$, respectively;
$\varphi(\underline{x},\underline{\pi})$ denotes a quantifier-free
formula where at most the variables in the tuple $\underline{x}$ may
occur free and at most the predicate symbols in the tuple
$\underline{\pi}$ may occur besides those of the signature over which
$\varphi$ is built.  In this paper, we consider only consistent
theories axiomatized by \emph{universal sentences} of the form
$\forall \underline{x}.\varphi(\underline{x})$.  In the examples, we
will make frequent use of the \emph{theory of scalar values $v_1, ...,
  v_n$ (for $n\geq 1$) of type $S$}, denoted with $SV(\{v_1, ...,
v_n\}, S)$, whose signature consists of the sort $S$, the constant
symbols $v_1, ..., v_n$ of sort $S$, and it is axiomatized by the
following (universal) sentences: $v_i \neq v_j$ for $i,j=1, ..., n$,
$i\neq j$, and $\forall x.(x=v_1 \vee \cdots \vee x=v_n)$, where $x$
is of sort $S$.

\subsection{Symbolic representation of RBAC policies} Let
$T_{\mathit{Role}}$ be a $\Sigma_{\mathit{Role}}$-theory axiomatized
by a finite set of universal sentences where $\Sigma_{\mathit{Role}}$
contains the sort $\mathit{Role}$, the predicate $\succeq$, and
countably many constants of sort $\mathit{Role}$ but no function
symbol.  Let $T_{\mathit{User}}$ be a $\Sigma_{\mathit{User}}$-theory
axiomatized by a finite set of universal sentences where
$\Sigma_{\mathit{User}}$ contains the sort $\mathit{User}$, countably
many constants of sort $\mathit{User}$ but no function symbol.  Let
$T_{\mathit{Permission}}$ be a $\Sigma_{\mathit{Permission}}$-theory
axiomatized by a finite set of universal sentences where
$\Sigma_{\mathit{Permission}}$ contains the sort $\mathit{Role}$ and
countably many constants of sort $\mathit{Permission}$ but no function
symbol.  We emphasize that the signatures of these three theories may
contain finitely many predicate symbols besides those mentioned above
but no function symbols.
\begin{example}
  \label{ex:zero}
  For the version of ARBAC we are considering, the theory
  $T_{\mathit{Role}}$ can be axiomatized by the following three
  universal sentences: $\forall r.(r\succeq r)$, $\forall
  r_1,r_2.((r_1\succeq r_2 \wedge r_2\succeq r_1) \Rightarrow
  r_1=r_2)$, and $\forall r_1,r_2,r_3.((r_1\succeq r_2 \wedge
  r_2\succeq r_3) \Rightarrow r_1\succeq r_3)$.  This means that
  $\succeq$ is interpreted as a partial order by the structures in
  $Mod(T_{\mathit{Role}})$.  The set of basic roles and their
  positions in the partial order can be defined, when considering
  Example~\ref{ex:intro}, as the following sentences: $Engineer
  \succeq Employee$, $PartTime \succeq Employee$, $FullTime \succeq
  Employee$, $ProjectLead \succeq Engineer$, and $Manager \succeq
  FullTime$.  
  %%% SR: COMMENT OR UNCOMMENT THE FOLLOWING LINES FOR THE SHORT OR LONG
  %%% VERSION, RESP.
  The interested reader can see~\cite{ext-version} for a discussion on
  how to formalize ARBAC with parametric roles.
  % In Appendix~\ref{app:extension-parametric-role}, we discuss how to
  % parametric roles can be formalized in our framework.
  %
  % For extensions of ARBAC using parametric roles, we can assume
  % $\Sigma_{\mathit{Role}}$ to contain also sort symbols for the
  % values of the parameter of each role and predicate symbols for the
  % parametric roles.  To illustrate, consider the university policy
  % example in~\cite{stoller2}.  There, the role schema
  % $\mathit{Student}(dept,cid)$ is used for students registered for
  % the course numbered $cid$ offered by department $dept$ and the
  % instance $\mathit{Student}(dept=cs,cid=101)$ identifies students
  % of the Computer Science department taking course $101$.  In our
  % framework, we can add the sort symbols $\mathit{Department}$ and
  % $\mathit{CourseId}$, and the predicate symbol $\mathit{Student}~:~
  % \mathit{Department}\times \mathit{CourseId}\times \mathit{Role}$
  % to $\Sigma_{\mathit{Role}}$.  Then, we can add a universal
  % sentence saying that the new predicate is partially functional,
  % i.e.\
  % \begin{eqnarray*}
  %   \forall D,CID,r_1,r_2.((\mathit{Student}(D,CID,r_1)\wedge 
  %   %                           \mathit{Student}(D,CID,r_2)) 
  %   %    & \Rightarrow & r_1 = r_2) ,
  %  \end{eqnarray*}
  % where $D$ is a variable of sort $\mathit{Department}$ and $CID$
  % is of sort $\mathit{CourseId}$. A more comprehensive discussion
  % about how to formalize parametric roles is included in
  % Appendix~\ref{app:extension-parametric-role}.

  For the theory $T_{\mathit{User}}$, we have a similar flexibility.
  For example, if there is only a \emph{finite and known} number
  $n\geq 1$ of users, say $e^u_1, ..., e^u_n$, then we can use the
  theory of a scalar value $SV(\{e^u_1, ..., e^u_n\}, \mathit{User})$.
  Another situation is when we have a \emph{finite but unknown} number
  of users whose identifiers are, for example, linearly ordered (think
  of the integers with the usual order relation `less than or equal').
  In this case, we add the ordering relation $\leq$ of arity
  $\mathit{User}\times \mathit{User}$ to $\Sigma_{\mathit{User}}$ and
  the following universal sentences constrain $\leq$ to be a linear
  order: $\forall u.(u\leq u)$, $\forall u_1,u_2,u_3.((u_1\leq
  u_2\wedge u_2\leq u_3)\Rightarrow u_1\leq u_3)$, $\forall
  u_1,u_2.((u_1\leq u_2\wedge u_2\leq u_1)\Rightarrow u_1 = u_2)$, and
  $\forall u_1,u_2.(u_1\leq u_2\vee u_2\leq u_1)$.  If
  $T_{\mathit{User}}=\emptyset$, then the identifiers $e^u_i$ of users
  can be compared for (dis-)equality % (recall that $=$ is a logical
%   symbol and it is thus available in every signature)
  and there is again a \emph{finite but unknown} number of users.

  Similar observations also hold for $T_{\mathit{Permission}}$.
  Often, there is only a finite and known number of permissions that
  can be associated to roles.  For example, continuing the
  formalization of Example~\ref{ex:intro}, recall that we have only
  three permissions: Access, View, and Edit.  So,
  $T_{\mathit{Permission}}:= SV(\{Access, View, Edit\},$
  $\mathit{Permission})$.  \qed
\end{example}
As shown by the example above, the flexibility of our approach allows
us to go beyond standard ARBAC policies by specifying the domains of
users, roles, and permissions enjoying non-trivial algebraic
properties which are useful to model, e.g., property-based
policies~\cite{uarbac}.  
% For example, property-based policies~\cite{uarbac} can be specified
% in our framework by enriching the theories with suitable axioms to
% describe the properties and their relationships with the user-role
% assignment.
We leave a detailed analysis of the scope of applicability of our
framework to future work (as a first step in this direction,
see~\cite{asiaccs}).

Now, we define $\Sigma_{ARBAC}:=\Sigma_{\mathit{Role}}\cup
\Sigma_{\mathit{User}}\cup \Sigma_{\mathit{Permission}}\cup \{ pa,
ua\}$ and let $T_{ARBAC} := T_{\mathit{Role}}\cup
T_{\mathit{User}}\cup T_{\mathit{Permission}}\cup \mathit{PA}$, where
$\mathit{PA}$ is a set of (universal) sentences over
$\Sigma_{\mathit{Role}}\cup \Sigma_{\mathit{Permission}}\cup \{pa\}$
characterizing the permission assignment relation.
\begin{example}
  \label{ex:zero-bis}
  Consider again
  Example~\ref{ex:intro}.  % From Example~\ref{ex:zero}, recall that
%   $T_{\mathit{User}} := SV(\{ Alice, Bob, Carol\}, \mathit{User})$,
%   $T_{\mathit{Role}} := SV(\{\mathit{Employee}$,
%   $\mathit{HumanResource}, \mathit{Engineer}\}, \mathit{Role})$ and
%   $T_{\mathit{Permission}} := SV(\{Access, View, Edit\}$,
%   $\mathit{Permission})$.  
  The permission-role assignment is axiomatized by
  $\mathit{PA}:=\{\forall p,r.(pa(p,r) \Leftrightarrow ( (p =
  \mathit{Access} \wedge r=\mathit{Employee}) \vee (p = \mathit{View}
  \wedge r=\mathit{HumanResource}) \vee (p = \mathit{Edit} \wedge
  r=\mathit{Engineer}) )\}$. % specifying that role
%   Employee has permission Access, role HumanResource has permission
%   View, and role Engineer has permission Edit.
  \qed
\end{example}
Observe that a structure in $Mod(T_{\mathit{ARBAC}})$ over
$\Sigma_{\mathit{ARBAC}}$ is a RBAC policy.  
% Notice that, without loss of generality, we can consider only those
% many-sorted $\Sigma_{ARBAC}$-structures in $Mod(T_{\mathit{ARBAC}})$
% whose universes are three disjoint sets of users, roles, and
% permissions.  This is so because although universes are not required
% to be disjoint in many-sorted first-order logic, equality is defined
% only on elements of the same universe; hence, there will always
% exist an elementary equivalent (i.e.\ satisfying the same set of
% formulae) structure whose universes are disjoint.
\subsection{Symbolic representation of initial RBAC policies,
  constraints, and goals} 
Since no axiom involving $ua$ is in $T_{\mathit{ARBAC}}$, the
interpretation of $ua$ is arbitrary.  We consider the problem of how
to constrain the interpretation of $ua$ by means of an example.
\begin{example}
  \label{ex:one} 
  We specify the user-role assignment of Example~\ref{ex:intro}.  Let
  $T_{\mathit{User}}, T_{\mathit{Role}}$, and
  $T_{\mathit{Permission}}$ be as in Example~\ref{ex:zero-bis}.
  Consider the formula
%   According to~\cite{li-tripunitara}, the user assignment relation
%   must be such that Alice is PartTime and Engineer, Bob is a Manager,
%   and Carol Is HumanResource.  This can be described by using the
%   following formula
  $In(ua)$:% \footnote{Notice that $In(ua)$ can be seen as the Clark's
%     completion of the following facts: $ua(Alice,PartTime)$,
%     $ua(Alice,Engineer)$, $ua(Bob,Manager)$, and
%     $ua(Carol,HumanResource)$.}
  \begin{eqnarray*}
    \forall u,r.(ua(u,r) & \Leftrightarrow &
    (
       (u=Alice \wedge r = PartTime)     \vee 
       (u=Alice \wedge r=Engineer)       \vee \\
       && \hspace{-1cm}
       (u=Bob \wedge r=Manager)          \vee 
       (u=Carol \wedge r=HumanResource) 
       ) ) .
  \end{eqnarray*}
  (Notice that $In(ua)$ can be seen as the Clark's
  completion~\cite{clark} of the facts: $ua(Alice,$ $PartTime)$,
  $ua(Alice,Engineer)$, $ua(Bob,Manager)$, and $ua(Carol,$
  $HumanResource)$.)  It is easy to see that the interpretation
  considered in Example~\ref{ex:intro} satisfies $In(ua)$.
%   \begin{eqnarray*}
%     s_{In}(ua) & := & \left\{  
%       \begin{array}{l}
%         (Alice, PartTime), (Alice, Engineer), \\
%         (Bob, Manager), (Carol,HumanResource) 
%       \end{array} \right\} % \\
%     v_{In}(pa) & := & \left\{  
%       \begin{array}{l}
%         (Access, Employee), (View, HumanResource), \\
%         (Edit, Engineer) 
%       \end{array} \right\} \\
%     v_{In}(rh) & := & \left\{  
%       \begin{array}{l}
%         (Engineer, Employee), (FullTime, Employee), \\
%         (PartTime, Employee), (ProjectLead, Engineer), \\
%         (Manager, FullTime) 
%       \end{array} \right\} 
%     . 
%   \end{eqnarray*}
%   Finally, notice that from $Inv_{rh}^{\succeq}$, we can easily
%   derive, e.g., that $PartTime \succeq Employee$ (since $rh[PartTime,
%   Employee]$) or that $ProjectLead \succeq Employee$ (since
%   $rh[PartTime, Engineer]$, $rh[Engineer,$ $Employee]$ and the
%   transitivity of $\succeq$).
%   In other words, the formula $In(ua)$ describes a (singleton) set of
%   valuation mapping for $ua$.
\qed
\end{example}
Since the formula $In(ua)$ used in the example above belongs to the
class of universal sentences containing the state variable $ua$, we
will use such a class of formulae, and denote it with
\emph{$\forall$-formulae}, to symbolically specify initial RBAC
policies.
\begin{example} 
  \label{ex:unbounded}
  Although in Example~\ref{ex:one} the numbers of users and roles are
  fixed to certain values, our framework does not require this.  For
  example, recall the discussion in Example~\ref{ex:zero} and take
  $T_{\mathit{User}}=\emptyset$, $T_{\mathit{Role}}=\emptyset$.  Then,
  consider the following $\forall$-formula: $\forall u,r.(ua(u,r)
  \Leftrightarrow (u\neq e^u_0 \wedge r\neq e^r_0))$.  A RBAC policy
  $s$ satisfying the formula is such that $(e^u_0,e^r_0)\not\in s(ua)$
  and $(e^u_i,e^r_j)\in s(ua)$ for every pair $(i,j)$ of natural
  numbers with $i,j\neq 0$.  Thus, there is no bound on the number of
  pairs $(e^u_i,e^r_j)$ in
  $s(ua)$.  % Our procedure for solving the user-role
%   reachability problem considers those initial RBAC policies whose
%   number of pairs $(e^u_i,e^r_j)$ with $i,j\neq 0$ is finite but
%   unknown, see Section~\ref{sec:reachability-analysis} for
%   details.
  \qed
\end{example}

Notice that $\forall$-formulae are not only useful to describe initial
RBAC policies but also to express constraints on the set of states
that \textsl{can\_assign} and \textsl{can\_revoke} actions must
satisfy.  As an example, consider RBAC policies with separation of
duty constraints, i.e.\ a user cannot be assigned two given roles.
This can be enforced by using \emph{static mutually exclusive roles}
(SMER) \emph{constraints} that require pairs of roles with disjoint
membership (see, e.g.,~\cite{stoller}).  Formulae representing SMER
constraints are $\forall$-formulae with the following form: $\forall
u.\neg (ua(u,e^r_i) \wedge ua(u,e^r_j) )$, for $i,j\geq 0$ and $i\neq
j$.  Notice that other kinds of constraints can be specified in our
framework as long as they can be expressed as $\forall$-formulae.

% Although $\forall$-formulae turn out to be quite useful, they cannot
% express the goal $\gamma$ of a user-role reachability problem as
% illustrated by the following example.
\begin{example}
  \label{ex:two}
  Let us consider again the situation described in
  Example~\ref{ex:intro}.  One may be interested in knowing if user
  Alice can take role FullTime and have permission Access.  This
  property can be encoded by the following formula:
  \begin{eqnarray*}
    \hspace{.725cm}
    \exists u,r,p.(ua(u,r) \wedge pa(p,r) \wedge
    u = Alice \wedge r \succeq FullTime \wedge p = Access) . 
    \hspace{.525cm} \qed
  \end{eqnarray*}
\end{example}
Generalizing this example, we introduce \emph{$\exists$-formulae} of
the form $\exists
\underline{u},\underline{r},\underline{p}.\varphi(\underline{u},\underline{r},\underline{p})$.

\subsection{Symbolic representation of administrative actions}
A \emph{policy literal} is either $ua(u,r)$, $\neg ua(u,r)$, a literal
over $\Sigma_{\mathit{User}}$ (e.g., $u=e^u_i$ or $u\neq e^u_i$ for
$i\geq 0$), or a literal over $\Sigma_{\mathit{Role}}$ (e.g.,
$r=e^r_j$, $r\succeq e^r_j$, or their negations for $j\geq 0$).  A
\emph{policy expression} is a finite conjunction of policy
literals.  % A policy expression $C$ is
% \emph{grounded} if for each atom of the form $ua(u,r)$ in $C$, there
% exist $u \bowtie e^u_i$ and $r \bowtie e^r_i$ that also occur in $C$
% for some $i\geq 1$ and $\bowtie\in\{ =, \neq\}$.  
Administrative actions are represented by instances of formulae of the
following form:% The formulae
% specifying \textsl{can\_assign} and \textsl{can\_revoke} actions have the following two
% forms, respectively:
\begin{eqnarray} 
  \label{eq:can_assign}
  \exists u,r, u_1, r_1, r_2, ..., r_k. & 
    (C(u,r, u_1, r_1, r_2, ..., r_k)  \wedge  & ua' = ua \oplus (u_1,e^r_i)) \\
  \label{eq:can_revoke}
  \exists u,r,u_1.&(C(u,r, u_1) ~~~~~~~~~~~~~~~~\wedge  & ua' = ua \ominus (u_1,e^r_i))
\end{eqnarray}
where $k,i\geq 0$, $C$ is a policy expression called the \emph{guard}
of the transition, primed variables denote the value of the state
variable $ua$ after the execution of the transition, $ua \odot
(u,e^r_i)$ abbreviates
\begin{eqnarray*}
  \lambda w,v.(\mathit{if}~(w=u\wedge v=e^r_i)
              ~\mathit{then}~b
              ~\mathit{else}~ua(w,v)) ,
\end{eqnarray*}
and $b$ is $\mathit{true}$ when $\odot$ is $\oplus$ and it is
$\mathit{false}$ when $\odot$ is $\ominus$.\footnote{We use
  $\lambda$-notation here for the sake of readability only.  The same
  formulae can be easily recast in pure first-order logic.  For
  example, (\ref{eq:can_assign}) can be written as $\exists u,r,r_1,
  ..., r_k.(C(u,r, r_1, ..., r_k) \wedge \forall w,r.(ua'(w,r)
  \Leftrightarrow ((w=u\wedge r=e^r)\vee ua(w,r)))$.} It is possible
to symbolically represent \textsl{can\_assign} actions as formulae of the form
(\ref{eq:can_assign}) and \textsl{can\_revoke} actions as formulae of the form
(\ref{eq:can_revoke}).  We illustrate this with an example.
% According to the current practice, administrative revocation has a
% tautological guard since there is little evidence that
% pre-conditions for revocation are useful.
% We illustrate how formulae of these form are useful in modelling `can
% assign and \textsl{can\_revoke} rules by continuing the example
% from~\cite{li-tripunitara} already considered in the examples above.
\begin{example}
  \label{ex:three} 
  We specify in our framework the administrative actions given in
  Example~\ref{ex:intro}.  The \textsl{can\_assign} action $\langle Manager,
  \{Engineer, FullTime\}, ProjectLead \rangle$ corresponds to the
  following instance of (\ref{eq:can_assign}):
  \begin{eqnarray*}
    \exists u,r,u_1,r_1,r_2.\left(
      \begin{array}{ll}
        ua(u,r) \wedge r\succeq Manager \wedge u\neq  Carol \wedge\\
        ua(u_1,r_1) \wedge r_1\succeq Engineer \wedge 
        ua(u_1,r_2) \wedge r_2\succeq FullTime  \wedge \\
        ua' = ua \oplus (u_1,ProjectLead) 
      \end{array} 
      \right) .
  \end{eqnarray*}
  Two observations are in order.  First, the literal $u\neq Carol$
  disables the transition when $u$ is instantiated to $Carol$.  This
  allows us to model the set $T=\{Carol\}$ of users that are prevented
  to execute assignments.  Second, by simple logical manipulations and
  recalling the definition of the abbreviation $ua^*$ introduced in
  Section~\ref{sec:ARBAC-intro}, it is possible to rewrite the guard
  of the transition as $ua^*(u,Manager)\wedge ua^*(u_1,Engineer)\wedge
  ua^*(u_1,FullTime)\wedge u\neq Carol$.  
  % This should convince the reader that the instance of
  % (\ref{eq:can_assign}) faithfully models the \textsl{can\_assign} action in
  % Example~\ref{ex:intro}.
  The simpler \textsl{can\_assign} rules $\langle HumanResource, \emptyset,
  FullTime \rangle$ and $\langle HumanResource$, $\emptyset,$
  $PartTime \rangle$ can be specified by the following two instances
  of (\ref{eq:can_assign}):
  \begin{eqnarray*}
    \exists u,r,u_1.\left(
      \begin{array}{ll}
        ua(u,r) \wedge r\succeq HumanResource \wedge u\neq  Carol & \wedge\\
        ua' = ua \oplus (u_1,FullTime)
      \end{array} 
      \right) \\
    \exists u,r,u_1.\left(
      \begin{array}{ll}
        ua(u,r) \wedge r\succeq HumanResource \wedge u\neq  Carol & \wedge\\
        ua' = ua \oplus (u_1,PartTime)
      \end{array} 
      \right) .
  \end{eqnarray*}
  Following~\cite{li-tripunitara}, we call $AATU$ (an abbreviation for
  `assignment and trusted users') the set containing the above three
  formulae.  

  The \textsl{can\_revoke} action $\langle Manager, ProjectLead \rangle$ is
  formalized by the following instance of (\ref{eq:can_revoke}):
  \begin{math}
    \exists u,r.(
        ua(u,r) \wedge r\succeq Manager  ~\wedge~
        ua' = ua \ominus (u_1,ProjectLead) ) .
  \end{math}
  The remaining three \textsl{can\_revoke}s can be obtained from the formula
  above by simply replacing Manager and ProjectLead with Manager and
  Engineer for $\langle Manager, Engineer \rangle$, with HumanResource
  and FullTime for $\langle HumanResource,$ $FullTime \rangle$, and
  with HumanResource and PartTime for $\langle HumanResource,$
  $PartTime \rangle$.  % As for \textsl{can\_assign} actions, the guard of the
%   instance of (\ref{eq:can_revoke}) above can be rewritten as
%   $ua^*(u,Manager)$.
%   \begin{eqnarray*}
%     \exists u,r,u_1,r_1.\left(
%       \begin{array}{ll}
%         ua(u,r) \wedge r\succeq Manager  & \wedge\\
%         r_1 = ProjectLead \wedge ua' = ua \ominus (u_1,r_1)
%       \end{array} 
%       \right) \\
%     \exists u,r,u_1,r_1.\left(
%       \begin{array}{ll}
%         ua(u,r) \wedge r\succeq Manager  & \wedge\\
%         r_1 = Engineer \wedge ua' = ua \ominus (u_1,r_1)
%       \end{array} 
%       \right) \\
%     \exists u,r,u_1,r_1.\left(
%       \begin{array}{ll}
%         ua(u,r] \wedge r\succeq HumanResource  & \wedge\\
%         r_1 = FullTime \wedge ua' = ua \ominus (u_1,r_1)
%       \end{array} 
%       \right) \\
%     \exists u,r,u_1,r_1.\left(
%       \begin{array}{ll}
%         ua(u,r] \wedge r\succeq HumanResource  & \wedge\\
%         r_1 = PartTime \wedge ua' = ua \ominus (u_1,r_1)
%       \end{array} 
%       \right) .
%   \end{eqnarray*}
%   As in~\cite{li-tripunitara}, let $AAR$ be the set obtained by adding
%   the last four formulae above to $AATU$. 
  \qed 
\end{example}
Notice that the guards of the transitions of the form
(\ref{eq:can_assign}) do not correspond exactly to those introduced in
Section~\ref{sec:ARBAC-intro}.  On the one hand, policy expressions
give us the possibility to require a user $u$ to be an explicit member
of a certain role $r$ in the guard of transition (by writing
$ua^*(u,r)$) while preconditions of a \textsl{can\_assign} can only
require a user to be an implicit member of a role (i.e.\ $ua^*(u,r)$).
On the other hand,
% it is possible to express $ua^*(u,r)$, 
it is not possible, in general, to express $\neg ua^*(u,r)$ (i.e.\ $u$
is neither an explicit nor an implicit member of $r$), although it is
possible to use $\neg ua(u,r)$ (i.e.\ $u$ is not an explicit member of
$r$).
% First, there is no disjunction.  This is without loss of generality as
% we can always transform a precondition to disjunctive normal form,
% i.e.\ a formula of the form $\bigvee_{i=1}^n C_i$ for $C_i$ a policy
% expression and then create as many instances of (\ref{eq:can_assign})
% as the disjuncts $C_1, ..., C_n$ with the same update for $ua$.  
This is so because to express $\neg ua^*(u,r)$, a universal
quantification is required; recall from Section~\ref{sec:ARBAC-intro}
that $\neg ua^*(u,r)$ abbreviates $\forall r'.(r'\succeq r\Rightarrow
\neg ua(u,r))$.  In other words, only a limited form of negation can
be expressed in the guards of our formalization of a \textsl{can\_assign}
action.  This simplifies the technical development that follows, in
particular the proof of termination of the procedure used to solve the
user-role reachability problem (see
Section~\ref{sec:reachability-analysis} for details).  We plan to
adapt a technique used in infinite state model checking for handling
global conditions to allow $\neg ua^*(u,r)$ in the guards of
transitions (see, e.g.,~\cite{abdulla-delzanno-rezine}) but leave this
to future work.  Here, we
% The second difference with respect to preconditions of \textsl{can\_assign}
% actions is that no universal quantifier may occur in guards; thus, it
% is impossible to express $\neg ua^*[u,r]$.  If one decides to allow
% universal quantifiers to occur in guards, then several nice properties
% of the framework we are proposing no more hold.  In particular, all
% those features that allow us to design a procedure capable of
% automatically solving the goal reachability problem are lost, e.g.,
% the satisfiability problems to which the goal reachability problem is
% reduced are no more decidable.  Here, we 
observe that in many situations of practical relevance, it is possible
to overcome this difficulty.  For example, when there are only
finitely many roles ranging over a set $R$, it is possible to
eliminate the hierarchy as explained in~\cite{stoller0} so that the
framework proposed in this paper applies without problems.
% the formula $\forall r'.(r'\succeq r\Rightarrow \neg ua(u,r))$ (whose
% abbreviation is $\neg ua^*(u,r)$) can be equivalently rewritten as
% $\bigwedge_{r'\in R} (r'\succeq r \Rightarrow \neg ua(u,r'))$ without
% using a quantifier.  Thus, any guard containing $\neg ua^*(u,r)$ is
% equivalent to a quantifier-free formula so that several instances of
% (\ref{eq:can_assign}) can be derived.  To understand how this is
% possible, recall that if the guard of a transition is a
% quantifier-free formula, it can be transformed to disjunctive normal,
% i.e.\ a formula of the form $\bigvee_{i=1}^n C_i$ for $C_i$ a policy
% expression.  At this point, it is possible to create as many instances
% of (\ref{eq:can_assign}) as the disjuncts $C_1, ..., C_n$ with the
% same update for $ua$.  
It is worth noticing that although the set of roles has been assumed
to be bounded, our framework supports the situation where the set of
users can be finite but its cardinality is unknown.  % The approaches

\subsection{Reachability and satisfiability modulo $T_{\mathit{ARBAC}}$}
\label{subsec:reach-and-smt}
At this point, it should be clear that the (algebraic) structures of
users, roles, and permission can be specified by suitable theories;
that we can symbolically represent RBAC policies and goals by using
$\forall$-formulae and $\exists$-formulae, respectively, \textsl{can\_assign}
actions by formulae of the form (\ref{eq:can_assign}), and \textsl{can\_revoke}
actions by formulae of the form (\ref{eq:can_revoke}).  As a
consequence, we can rephrase the user-goal reachability problem
introduced in Section~\ref{sec:ARBAC-intro} as follows.

Let $T_{\mathit{ARBAC}}$ be a $\Sigma_{\mathit{ARBAC}}$-theory given
as described above and specifying the structure of users, roles,
permission, role hierarchy, and the permission-role relation.  If
$\Gamma:=(S_0, A)$ is an ARBAC policy together with a set
$\mathcal{C}$ of constraints on the set of states that the actions of
the system must satisfy (e.g., SMER), then derive the associated
\emph{symbolic ARBAC policy} $\Gamma_s:=(In(ua), Tr, C)$ as explained
above, where $In$ is a $\forall$-formula representing the initial set
$S_0$ of RBAC policies, $Tr$ is a finite set of instances of
(\ref{eq:can_assign}) or of (\ref{eq:can_revoke}) corresponding to the
actions in $A$, and $C$ is a finite set of $\forall$-formula
representing constraints in $\mathcal{C}$.
Furthermore, let $\gamma_s$ be an $\exists$-formula of the form
\begin{eqnarray}
  \label{eq:symbolic-goal}
  \exists u_1, r_1, p_1, ... , u_n, r_n, p_n.
   \bigwedge_{i=1}^n (ua(u_i,r_i) \wedge 
                     r_i \bowtie e^r_{j_i} \wedge
                     p_i=e^p_{j_i}) ,
\end{eqnarray}
called a \emph{symbolic goal} and corresponding to a goal $RP:= \{
(e^r_{j_i},e^p_{j_i}) ~|~ i= 1, ..., n\}$, where $\bowtie\in \{ =,
\succeq\}$.  Then, it is easy to see that the user-role reachability
problem for $\Gamma$ with $RP$ as goal is solvable iff there exists a
natural number $\ell\geq 0$ such that the formula
\begin{equation}
  \label{eq:unsafe}
  In(ua_0)\wedge 
  \bigwedge_{i=0}^{\ell} (\iota({a}_i) \wedge 
                        \tau(ua_i, ua_{i+1})\wedge
                        \iota({a}_{i+1})) 
  \wedge \gamma_s(ua_{\ell})
\end{equation}
is {satisfiable modulo $T_{ARBAC}$}, where $\tau$ is the disjunction
of the formulae in $Tr$, and $\iota$ is the disjunction of those in
$C$.  Notice that the (big) conjunction over $\ell$ with $In$ in
(\ref{eq:unsafe}) can be seen as a characterization of the set of
states (forward) reachable from the initial set of states.
Symmetrically (and more interestingly for the rest of this paper), the
(big) conjunction over $\ell$ with $\gamma_s$ in (\ref{eq:unsafe})
characterizes the set of states backward reachable from the goal
states.
% An instance of the \emph{goal reachability problem} is a pair
% $(\Gamma, \gamma)$ and consists of establishing whether there exists a
% natural number $\ell\geq 0$ such that the formula
% \begin{equation}
%   \label{eq:unsafe}
%   In(ua_0)\wedge 
%   \bigwedge_{i=0}^{\ell} (\iota({a}_i) \wedge 
%                         \tau(ua_i, ua_{i+1})\wedge
%                         \iota({a}_{i+1})) 
%   \wedge \gamma(ua_{\ell})
% \end{equation}
% is \emph{satisfiable modulo the theory $T_{URA}$}, i.e.\ there exists
% a model of $T_{URA}$ satisfying (\ref{eq:unsafe}), where
% $\tau:=\bigvee_{i=1}^n \tau_i$, $\iota:=\bigwedge_{i=1}^m \iota_i$,
% and $ua_0, ua_1, ..., ua_{\ell}$ are renamed copies of the state
% variable $ua$.  
% If there is no such $\ell$, then the goal $\gamma$ is
% \emph{unreachable}; otherwise, it is \emph{reachable} since the
% satisfiability of \eqref{eq:unsafe} implies the existence of a
% sequence of transitions of length $\ell$, called a \emph{run}, leading
% $\Gamma_s$ from a state in $In$ to a state satisfying $\gamma_s$.  
We observe that when $\ell=0$, no actions must be performed and
already some of the states in $In$ satisfies $\gamma_s$, thus, formula
(\ref{eq:unsafe}) simplifies to $In(ua_0)\wedge \iota(ua_0)\wedge
\gamma_s(ua_0)$.  
\begin{example} 
  \label{ex:role-I-axioms}
  We illustrate the check for satisfiability of the formula
  (\ref{eq:unsafe}) for $\ell=0$ by reconsidering the situation
  described in Example~\ref{ex:two}.  The problem was to establish if
  the formula $In(ua)$ of Example~\ref{ex:one} and the goal formula
%   \begin{eqnarray*}
%     \exists u,r,p.(ua(u,r)] \wedge pa[p,r] \wedge
%     u = Alice \wedge r \succeq FullTime \wedge p = Access) 
%   \end{eqnarray*}
  of Example~\ref{ex:two} are satisfiable modulo the theory
  $T_{\mathit{ARBAC}}$ in Example~\ref{ex:zero-bis}.  % Since we are
%   considering the situation where no administrative actions can be
%   taken, we can trasform the problem into a particular instance of a
%   reachability problem under the assumptions of URA97.  In fact, we
%   can see the first conjunct of $In(ua,pa,rh)$ in Example~\ref{ex:one}
%   as $In(ua)$, the singleton set containing its second conjunct as
%   $I_{Permission}$ (i.e.\ an axiom of the theory $T_{Permission}$),
%   and the signelton set containing its third conjuct (in which, we
%   replace $rh$ with $\succeq$) as $I_{Role}$ (i.e.\ an axioms of the
%   theory $T_{Role}$).  
  We assume that the set of constraints of the symbolic ARBAC polices
  is empty.  In this context, the formula (\ref{eq:unsafe}) above can
  be written as follows: 
  \begin{eqnarray*}
    PO ~ := ~ \forall u,r.(ua(u,r) \Leftrightarrow 
    \left(
     \begin{array}{ll}
       (u=Alice \wedge r = PartTime)    &  \vee \\
       (u=Alice \wedge r=Engineer)      &  \vee \\
       (u=Bob \wedge r=Manager)         &  \vee \\
       (u=Carol \wedge r=HumanResource) &
     \end{array}
     \right) )  & \wedge \\
    \exists u_1,r_1,p_1.(ua(u_1,r_1) \wedge pa(p_1,r_1) \wedge
                   u_1 = Alice \wedge r_1 \succeq FullTime \wedge p_ = Access)  & ,
   \end{eqnarray*}
   where the existentially quantified variables in the goal have been
   renamed for clarity.  The problem is to establish the
   satisfiability of $PO$ modulo the theory $T_{\mathit{ARBAC}}$ in
   Example~\ref{ex:zero-bis}.  As it will be seen below, there exists
   an algorithm capable of answering this question automatically.  For
   $PO$, the algorithm would return `unsatisfiable,' entitling us to
   conclude that the set of initial states considered in
   Example~\ref{ex:one} do not satisfy the goal of allowing Alice, who
   is a full-time employee, to get access to a certain resource. \qed
\end{example}
% Two remarks are in order.  First, let a $\Gamma:=(s_0, \{\alpha_1,
% ..., \alpha_n\})$ be an ARBAC system, $\Gamma_s:=(In(ua), \{\tau_1,
% ..., \tau_n\})$ be a symbolic ARBAC system (we have omitted the
% constraint for simplicity but the following discussion can be easily
% extended to also consider them), $RP:=\{(e^r_j,e^p_k)\}$ be a goal
% (the generalization to several pairs of roles and permissions is
% straightforward), and $\gamma_s$ be a symbolic goal of the form
% $\exists u,r,p.(ua(u,r]\wedge r \bowtie e^r_j \wedge p=e^p_k)$.  If
% $s_0\models In(ua)$, for every pair of states $s,s'$ we have that
% $(s,s')\in alpha_i$ iff $s,s'\models \tau_i(ua,ua')$, then the goal
% reachability problem is solvable iff the symbolic reachability problem
% is so.  The second remark is the following.  
If we were able to automatically check the satisfiability of formulae
of the form (\ref{eq:unsafe}), an idea to solve the user-role
reachability problem for ARBAC policies would be to generate instances
of (\ref{eq:unsafe}) for increasing values of $\ell$.  However, this
would not give us a decision procedure for solving the goal
reachability problem but only a semi-decision procedure.  In fact, the
method terminates only when the goal is reachable from the initial
state, i.e.\ when, for a certain value of $\ell$, the instance of the
formula (\ref{eq:unsafe}) is unsatisfiable modulo
$T_{\mathit{ARBAC}}$.  When, instead, the goal is not reachable, the
check will never detect the unsatisfiability and we will be forced to
generate an infinite sequence of instances of (\ref{eq:unsafe}) for
increasing values of $\ell$.  In other words, the decidability of the
satisfiability of (\ref{eq:unsafe}) modulo $T_{\mathit{ARBAC}}$ is
only a necessary condition for ensuring the decidability of the
user-role reachability problem.
% In fact, when the satisfiability check is negative---i.e.\
% (\ref{eq:unsafe}) is unsatisfiable modulo $T_{\mathit{ARBAC}}$---we
% do not known when to stop considering larger and larger values of
% $\ell$ so as to conclude that the goal is unreachable.
Fortunately, is possible to stop enumerating instances of
(\ref{eq:unsafe}) for a certain value $\overline{\ell}$ of $\ell$ when
the formula characterizing the set of reachable states for
$\ell=\overline{\ell}+1$ implies that characterizing the set of
reachable states for $\ell=\overline{\ell}$; i.e.\ we have detected a
\emph{fixed-point}.  We explore this idea in the following section.

\section{Symbolic analysis of ARBAC policies} %
\label{sec:reachability-analysis}
A general approach to solve the user-role reachability problem is
based on computing the set of backward reachable states.  It is
well-known that the computation of sets of backward (rather than
forward) reachable states is easier to mechanize.  For $n\geq 0$, the
$n$-\emph{pre-image} of a formula $K(ua)$ is a formula $Pre^n(\tau,K)$
recursively defined as follows: $Pre^0(\tau,K) := K$ and
$Pre^{n+1}(\tau, K) := Pre(\tau, Pre^n(\tau, K))$, where\footnote{In
  (\ref{eq:def-pre}), we use a second order quantifier over the
  relation symbol $ua$, representing the state of the system.  This
  should not worry the reader expert in first-order theorem proving
  since a higher-order feature is only used to give the definition of
  pre-image.  We will see that we can compute a first-order formula
  logically equivalent to (\ref{eq:def-pre}) so that only first-order
  techniques should be used to mechanize our approach.}
\begin{eqnarray}
  \label{eq:def-pre}
  Pre(\tau, K) & := & \exists ua'.(\tau(ua,ua') \wedge K(ua'))  .
\end{eqnarray}
The formula $Pre^n(\tau, \gamma)$ describes the set of states from
which it is possible to reach the goal $\gamma$ in $n\geq 0$ steps.
\begin{figure}[tb]
  \begin{center}
  \begin{tabular}{c}
    \begin{minipage}{.45\textwidth}
      \begin{tabbing}
        foo \= foo \= \kill
        \textbf{function} $\mathsf{BReach}(\Gamma ~:~ (In,Tr,C),~ \gamma ~:~ \exists\mbox{-formula})$ \\
        1 \> $P\longleftarrow \gamma$;  $B\longleftarrow \mathit{false}$; $\tau\longleftarrow \bigvee_{t\in Tr} t$; $\iota\longleftarrow \bigwedge_{i\in C} i$;\\
        2\> \textbf{while} ($\iota\wedge P\wedge \neg B$ is satisfiable modulo $T_{\mathit{ARBAC}}$) \textbf{do}\\
        3\>\> \textbf{if} ($In\wedge P$ is satisfiable modulo $T_{\mathit{ARBAC}}$) \\
        \>\> \hspace{.75cm} \textbf{then return}  $\mathsf{reachable}$;\\
        4\> \> $B\longleftarrow P\vee B$; \\ 
        5\>\> $P\longleftarrow Pre(\tau, P);$ \\
        6\> \textbf{end} \\
        7\> \textbf{return} $\mathsf{unreachable};$ 
      \end{tabbing}
    \end{minipage}
  \end{tabular}
  \end{center}
  \caption{\label{fig:reach-algo}%
    The basic backward reachability procedure}
\end{figure}
At the $n$-th iteration of the loop, the \emph{backward reachability
  algorithm} depicted in Figure~\ref{fig:reach-algo}, stores the
formula $Pre^n(\tau, \gamma)$ in the variable $P$ and the formula
$BR^n(\tau, \gamma):=\bigvee^n_{i=0} Pre^i(\tau, \gamma)$
(representing the set of states from which the goal $\gamma$ is
reachable in \emph{at most} $n$ steps) in the variable $B$.  While
computing $BR^n(\tau, \gamma)$, $\mathsf{BReach}$ also checks whether
the goal is reachable in $n$ steps (cf.\ line 3, which can be read as
$In\wedge Pre^n(\tau, \gamma)$ is satisfiable modulo
$T_{\mathit{ARBAC}}$) or a fixed-point has been reached (cf.\ line 2,
which can be read as $\neg((\iota\wedge BR^n(\tau, \gamma))\Rightarrow
BR^{n-1}(\tau, \gamma))$ is unsatisfiable modulo $T_{\mathit{ARBAC}}$
or, equivalently, that $((\iota\wedge BR^n(\tau, \gamma))\Rightarrow
BR^{n-1}(\tau, \gamma))$ is valid modulo $T_{\mathit{ARBAC}}$).
{Notice that $BR^{n-1}(\tau, \gamma)\Rightarrow BR^n(\tau, \gamma)$ is
  valid by construction; thus, if $((\iota\wedge BR^n(\tau,
  \gamma))\Rightarrow BR^{n-1}(\tau, \gamma))$ is a logical
  consequence of $T_{ARBAC}$, then also $((\iota\wedge BR^n(\tau,
  \gamma))\Leftrightarrow BR^{n-1}(\tau, \gamma))$ is so and a
  fixed-point has been reached.}  The invariant $\iota$ is conjoined to
the set of backward reachable states when performing the fixed-point
check as only those states that also satisfies the constraints are
required to be considered.  When $\mathsf{BReach}$ returns
$\mathsf{unreachable}$ (cf.\ line 7), the variable $B$ stores the
formula describing the set of states which are backward reachable from
$\gamma$ which is also a fixed-point.  Otherwise, when it returns
$\mathsf{reachable}$ (cf.\ line 3) at the $n$-th iteration, there
exists a run of length $n$ that leads the ARBAC policy from a RBAC
policy in $In$ to one in $\gamma$.
%%% SR dropped for the moment... %%%%%%%%%%%%%%%%%%%%%%%%%%%%%%%%%%%%%
% With some further book-keeping, it is possible to return the exact
% sequence of transitions in run when exiting with
% $\mathsf{reachable}$.  Thus, $\mathsf{BReach}$ can be used as the
% basis for an automated analysis techniques capable of verifying the
% correctness of policies but also for their debugging.  Indeed, this
% is an interesting feature of the approach because error-finding
% nicely complements verification. 
%%%%%%%%%%%%%%%%%%%%%%%%%%%%%%%%%%%%%%%%%%%%%%%%%%%%%%%%%%%%%%%%%%%%%%
We observe that for $\mathsf{BReach}$ to be an effective (possibly
non-terminating) procedure, it is mandatory that (i) the formulae used
to describe the set of backward reachable states are closed under
pre-image computation and (ii) both the satisfiability test for safety
(line 3) and that for fixed-point (line 2) are effective.

Regarding (i), it is sufficient to prove the following result.
\begin{property}
  \label{prop:closure-pre}
  Let $K$ be an $\exists$-formula.  If $\tau$ is of the form
  (\ref{eq:can_assign}) or (\ref{eq:can_revoke}), then $Pre(\tau,K)$
  is equivalent (modulo $T_{\mathit{ARBAC}}$) to an effectively
  computable $\exists$-formula.
\end{property}
\begin{proof}
  Let $K(ua):=\exists \underline{\tilde{u}},
  \underline{\tilde{r}}.\gamma(\underline{\tilde{u}},
  \underline{\tilde{r}},ua(\underline{\tilde{u}},\underline{\tilde{r}}))$,
  where $\gamma$ is a quantifier-free formula.  By definition,
  $Pre(\tau,K)$ is $\exists ua'.(\tau(ua,ua') \wedge K(ua'))$ and
  there are two cases to consider.  The former is when $\tau$ is of
  the form (\ref{eq:can_assign}).  In this case, $\exists
  ua'.(\tau(ua,ua') \wedge K(ua'))$ is equivalent to
  \begin{eqnarray*}
      \exists u,r, u_1, r_1, r_2, ..., r_k. 
      \left(
        \begin{array}{l}
          C(u,r, u_1, r_1, r_2, ..., r_k)  \wedge  \\
          \exists \underline{\tilde{u}},
          \underline{\tilde{r}}.\gamma(\underline{\tilde{u}},
          \underline{\tilde{r}},(ua \oplus (u_1,e^r))(\underline{\tilde{u}},\underline{\tilde{r}}))
        \end{array}
        \right)
  \end{eqnarray*}
  by simple logical manipulations and recalling the definition of $K$.
  In turn, this can be expanded to
  \begin{eqnarray*}
    \begin{array}{l}
      \exists u,r, u_1, r_1, r_2, ..., r_k.
      (C(u,r, u_1, r_1, r_2, ..., r_k)  \wedge  \\
      ~~~~\exists \underline{\tilde{u}},
      \underline{\tilde{r}}.\gamma(\underline{\tilde{u}},
      \underline{\tilde{r}},
      (  \lambda w,r.(\mathit{if}~(w=u\wedge r=e^r)
      ~\mathit{then}~\mathit{true}
      ~\mathit{else}~ua(w,r))) (\underline{\tilde{u}},\underline{\tilde{r}})))
    \end{array}
  \end{eqnarray*}
  by recalling the definition of $\oplus$.  It is possible to
  eliminate the $\lambda$-expression by observing that each of its
  occurrence will be applied to a pair of existentially quantified
  variables from $\underline{\tilde{u}},\underline{\tilde{r}}$ so that
  $\beta$-reduction can be applied.  After this phase, the
  `if-then-else' expressions can be eliminated by using a simple
  case-analysis followed by the moving out of the existential
  quantifiers that allows us to obtain an $\exists$-formula.  This
  concludes the proof of this case.  The second case, i.e.\ when
  $\tau$ is of the form (\ref{eq:can_revoke}), is omitted because
  almost identical to the previous.  \qed
\end{proof}
Observe also that $Pre(\bigvee_{i=1}^n\tau_i, K)$ is equivalent to
$\bigvee_{i=1}^n Pre(\tau_i, K)$ for $\tau_i$ of forms
(\ref{eq:can_assign}) and (\ref{eq:can_revoke}), for $i=1,...,n$.
\begin{example}
  \label{ex:four}
  To illustrate Property~\ref{prop:closure-pre}, we consider one of
  the transitions written in Example~\ref{ex:three} and the goal in
  Example~\ref{ex:two}.  We compute the pre-image w.r.t.\ the second
  transition in $AATU$ (where $HR$ stands for $HumanResource$ and $FT$
  for $FullTime$), i.e.
  \begin{eqnarray*}
    \exists u,r,p.(ua'(u,r) \wedge pa(p,r) \wedge
                   u = Alice \wedge r \succeq FT \wedge
                   p = Access) \wedge \\
    \exists u_1,r_1,u_2.(
        ua(u_1,r_1) \wedge r_1=HR \wedge u_1\neq  Carol  \wedge 
        ua' = ua \oplus (u_2,FT) ) & ,
  \end{eqnarray*}
  where $ua'$ is implicitly existentially quantified.  By simple
  logical manipulations, we have
  \begin{eqnarray*}
    \exists u,r,p, u_1,r_1,u_2 .(pa(p,r)\wedge
     (\mathit{if}~u=u_2\wedge r=FT
      ~\mathit{then}~\mathit{true}
      ~\mathit{else}~ ua(u,r)) \wedge \\
      u = Alice \wedge r \succeq FT \wedge p = Access \wedge 
      ua(u_1,r_1) \wedge r_1=HR \wedge u_1\neq  Carol ) ,
  \end{eqnarray*}
  which, by case analysis and some simplification steps, can be
  rewritten to
  \begin{eqnarray*}
    \exists u,r,p, u_1,r_1,u_2 .(pa(p,r)\wedge
    (r=FT \wedge u_2 = Alice  \wedge p = Access & \wedge \\
    ua(u_1,r_1) \wedge r_1=HR \wedge u_1\neq  Carol) & \vee \\
    (pa(p,r)\wedge u\neq u_2\wedge ua(u,r) \wedge u = Alice \wedge r \succeq FT
    \wedge  p = Access & \wedge \\
    ua(u_1,r_1) \wedge r_1=HR \wedge u_1\neq  Carol) & \vee \\
    (pa(p,r)\wedge r\neq FT \wedge ua(u,r) \wedge 
    u = Alice \wedge r \succeq FT \wedge p = Access & \wedge \\
    ua(u_1,r_1) \wedge r_1=HR \wedge u_1\neq  Carol ) ) & ,
  \end{eqnarray*}
  which is an $\exists$-formula according to
  Property~\ref{prop:closure-pre}.  \qed
\end{example}
Concerning the decidability of the satisfiability tests for safety and
fixed-point in the backward reachability algorithm in
Figure~\ref{fig:reach-algo} (point (ii) above), we observe that the
formulae at lines 2 and 3 can be effectively transformed to formulae
in the form $\exists
\underline{x}\forall\underline{y}.\varphi(\underline{x},
\underline{y}, ua)$ where $\underline{x}$ and $\underline{y}$ are
disjoint, which belong to the \emph{Bernays-Sch\"onfinkel-Ramsey}
(BSR) class (see, e.g.,~\cite{demoura-bjoerner-piskac}).  To see how
this is possible, let us consider the formulae at line 2.  This is the
conjunction of a $\forall$-formula ($\iota$), an $\exists$-formula (as
discussed above, the variable $P$ stores $Pre^n(\tau,\gamma)$, which
by Property~\ref{prop:closure-pre} is an $\exists$-formula), and
another $\forall$-formula (as discussed above, the variable $B$ stores
$\bigvee_{i=0}^n Pre^i(\tau,\gamma)$ whose negation is a conjunction
of $\forall$-formulae by Property~\ref{prop:closure-pre}, which is a
$\forall$-formula).  By moving out quantifiers (which is always
possible as quantified variables can be suitably renamed), it is
straightforward to obtain a BSR formula.  Now, let us turn our
attention to the formula at line 3.  It is obtained by conjoining a
$\forall$-formula ($In$ is so by assumption) and an $\exists$-formula
(stored in the variable $P$, see previous case).  Again, by simple
logical manipulations, it is not difficult to obtain a formula in the
BSR class.  We also observe that checking the satisfiability of BSR
formulae modulo $T_{\mathit{ARBAC}}$ can be reduced to checking the
satisfiability of formulae in the BSR class since all the axioms of
$T_{\mathit{ARBAC}}$ are universal sentences, i.e.\ BSR formulae.
Collecting all these observations, we can state the following
result.% \footnote{One may observe that there infinitely many
%   disequalities in the axioms of $T_{\mathit{ARBAC}}$ and this may
%   give rise to undecidability of the satisfiability checks.  Indeed,
%   this is not the case because only finitely many such disequalities
%   must be considered at any time.  In particular, only those
%   disequalities involving the constants occurring in $\varphi$ need to
%   be considered.}
\begin{property}
  \label{lem:dec-URA}
  The satisfiability tests at lines 2 and 3 of the backward
  reachability procedure in Figure~\ref{fig:reach-algo} are decidable.
\end{property}
This property is a corollary of the decidability of the satisfiability
of the BSR class (see, e.g.,~\cite{demoura-bjoerner-piskac}).
% There are many ways to prove this result.  Probably, the easiest way
% is to apply Herbrand theorem and realize that the domain of the
% Herbrand interpretation is finite.  More precisely, we Skolemize the
% existential prefix and obtain the formula $\forall
% \underline{y}.\varphi(\underline{x}, \underline{y})$, where the
% variables in $\underline{x}$ are considered as ``fresh'' constants.
% Then, the Herbrand domain is a finite set of constants in
% $\underline{x}$ and those occurring in $\varphi$, say $C$.  Following
% Herbrand instantiation procedure, it is possible to obtain the
% following ground formula to be checked for satisfiability:
% $\bigwedge_{\sigma} \varphi(\underline{x}, \underline{y}\sigma)$,
% where $\sigma$ ranges over all possible instantiations of the
% variables in $\underline{y}$ to the constants in $\underline{x}\cup
% C$.  The satisfiability of this formula can be checked by using a
% satisfiability (SAT) solver for proposition logic by considering all
% the atoms in the instances of $\varphi$ as propositional
% letters. Indeed, since it is required to consider possibly huge
% numbers of instances, heuristics to perform these instances
% efficiently are needed (see, e.g.,~\cite{demoura-bjoerner-piskac}).
Example~\ref{ex:four} above contains an illustration of a
satisfiability test to which Property~\ref{lem:dec-URA} applies.

\subsection{Termination}
%\label{subsec:br-termination}
The closure under pre-image computation
(Property~\ref{prop:closure-pre}) and the decidability of the
satisfiability checks (Property~\ref{lem:dec-URA}) guarantee the
possibility to mechanize the backward reachability procedure in
Figure~\ref{fig:reach-algo} but do not eliminate the risk of
non-termination.  There are various sources of diverge.  For example,
the existential prefix of a pre-image is extended at each pre-image
computation with new variables as shown in the proof of
Property~\ref{prop:closure-pre}.  % This potential source
% of divergence has been already discussed in~\cite{stoller2} for a
% similar problem.  
Another potential problem is that the fixed-point could not be
expressed by using disjunctions of $\exists$-formulae (according to
line 4 in Figure~\ref{fig:reach-algo}) even if it exists so that the
procedure is only able to compute approximations and thus never
terminates.
% (this was also noted in~\cite{stoller2} where it is said that ``states
% are described by potentially unbounded sets of parameterized boolean
% variables corresponding to role membership facts.''  
To show that both problems can be avoided and that the procedure in
Figure~\ref{fig:reach-algo} terminates,
% (thus proving the decidability of the goal reachability problem)
we follow the approach proposed in~\cite{ijcar08,AbdullaTCS} for
proving the termination of backward reachability for certain classes
of infinite state systems.  We introduce a model-theoretic notion of
certain sets of states, called \emph{configurations}, which are the
semantic counter-part of $\exists$-formulae, and then define a
well-quasi-order on them: this, according to the results
in~\cite{AbdullaTCS}, implies the termination of the backward
reachability procedure.  For lack of space, the full technical
development is omitted and can be found in
%%% SR: COMMENT OR UNCOMMENT THE FOLLOWING LINES FOR THE SHORT OR LONG
%%% VERSION, RESP.
\cite{ext-version};
% Appendix~\ref{app:termination};
here, we only sketch the main ideas.  We also point out that this
result can be seen as a special case of that in~\cite{ijcar08},
developed in a more general framework that allows for the
formalization and the analysis of safety properties for concurrent,
distributed, and timed systems as well as algorithms manipulating
arrays.  However, we believe worthwhile to prove termination for the
procedure presented in this paper (along the lines of~\cite{ijcar08})
as some technical definitions become much simpler.

A \emph{state of the symbolic ARBAC policy $\Gamma:= (In, Tr, C)$} is
a structure $\mathcal{M}\in Mod(T_{ARBAC})$, i.e.\ it is an RBAC
policy belonging to a certain class of first-order structures.  A
\emph{configuration} of $\Gamma$ is a state $\mathcal{M}$ such that
the cardinality of the domain of $\mathcal{M}$ is finite.
% For the sake of simplicity, $\mathcal{M}$ will be omitted whenever
% clear from the context.
Intuitively, a configuration is a finite representation of a possibly
infinite set of states that ``contains at least the part mentioned in
the configuration.''  The following example can help to grasp the
underlying intuition.
\begin{example}
  \label{ex:goal-semantics}
  As in Example~\ref{ex:unbounded}, let $T_{\mathit{User}}=\emptyset$,
  $T_{\mathit{Role}}=\emptyset$.  Consider the $\exists$-formula:
  $\exists u,r.(ua(u,r) \wedge u=e^u_0 \wedge r=e^r_0)$.  There is no
  bound on the number of pairs $(e^u_i,e^r_k)$ in a RBAC policy $s$
  satisfying the $\exists$-formula above provided that
  $(e^u_0,e^r_0)\in s(ua)$.  Our procedure for the reachability
  problem considers (only) those RBAC policies $s$ of the form
  $s(ua)=\{(e^u_0,e^r_0)\} \cup \Delta$ where $\Delta$ is a (possibly
  empty) set of pairs $(e^u_{i},e^r_{k})$ with $i,j\neq 0$.  In other
  words, the procedure considers all those configurations which
  contain at least the pair $(e^u_0,e^r_0)$ mentioned in the
  $\exists$-formula above plus any other (finite) set $\Delta$ of
  pairs.\qed
\end{example}
The idea that a configuration represents a (possibly infinite) set of
RBAC policies sharing a common (finite) set of user-role assignments
can be made precise by using the notion of partial order.  A
\emph{pre-order} $(P,\leq)$ is the set $P$ endowed with a reflexive
and transitive relation.  An \emph{upward closed set $U$} of the
pre-order $(P,\leq)$ is such that $U\subseteq P$ and if $p\in U$ and
$p\leq q$ then $q\in U$.  A \emph{cone} is an upward closed set of the
form $\uparrow p = \{ q\in P ~|~ p\leq q\}$.  We define a
\emph{pre-order on configurations} as follows.  Let $\mathcal{M}$ and
$\mathcal{M}'$ be configurations of $\Gamma$; $\mathcal{M} \leq
\mathcal{M}'$ iff there exists an embedding from $\mathcal{M}$ to
$\mathcal{M}'$.  Roughly, an embedding is a homomorphism that
preserves and reflects relations (see~\cite{ext-version} for a formal
definition)
%%% SR: UNCOMMENT THE FOLLOWING LINE FOR THE  LONG VERSION.
% (for a precise definition see Appendix~\ref{app:termination})
.  A configuration is the semantic counter-part of an
$\exists$-formula.  Let $[[K]] := \{ \mathcal{M}\in
Mod(T_{\mathit{ARBAC}}) ~|~ \mathcal{M} \models K \}$, where $K$ is an
$\exists$-formula.
\begin{lemma}
  \label{lem:semantic-counterpart}
  The following facts hold: (i) for every $\exists$-formula $K$, the
  set $[[K]]$ is upward closed and (ii) $[[K_1]]\subseteq [[K_2]]$ iff
  $(K_1\Rightarrow K_2)$ is valid modulo $T_{\mathit{ARBAC}}$, for
  every pair of $\exists$-formulae $K_1,K_2$.
\end{lemma}
An upward closed set $U$ is \emph{finitely generated} iff it is a
finite union of cones.  A pre-order $(P,\leq)$ is a
\emph{well-quasi-ordering (wqo)} iff every upward closed sets of $P$
is finitely generated. This is equivalent to the standard definition
of wqo, see~\cite{ijcar08} for a proof.  The idea is to use only
finitely generated upward closed sets as configurations so that their
union is also finitely generated and we can conclude that the backward
reachability procedure in Figure~\ref{fig:reach-algo} is terminating
because of the duality between configurations and $\exists$-formulae
(Lemma~\ref{lem:semantic-counterpart}).% and the procedure computes a
% fixed-point by storing a disjunction of $\exists$-formulae in the
% variable $B$ or, dually, a union of configurations which must be
% finite because of the wqo on configurations.
\begin{theorem}
  \label{thm:termination}
  The backward reachability procedure in Figure~\ref{fig:reach-algo}
  terminates.
\end{theorem}
As a corollary, we immediately obtain the following
fact.
% decidability of the user-role reachability problem for a theory
% $T_{\mathit{ARBAC}}$ specified as described in
% Section~\ref{sec:symbolic-specification-ARBAC}.
\begin{theorem}
  \label{thm:goal-reach-dec}
  The user-role reachability problem is decidable.
\end{theorem} 
This result is more general that those
in~\cite{li-tripunitara,stoller} which assume a bounded number of
users and roles.  We are more general than~\cite{stoller2} in allowing
for a finite but unknown number of users and roles while
in~\cite{stoller2} the users are bounded and only the roles are
parametric.  However, we allow for only a restricted form of negation
in the preconditions of \textsl{can\_assign} actions while~\cite{stoller2}
seems to allow for arbitrary negation.  Moreover, our procedure can
consider several initial RBAC policies at the same time
while~\cite{li-tripunitara,stoller,stoller2} can handle only one.

% This result generalizes those available the literature in several
% respects.  First, the results of our procedure holds for a finite
% number of users and roles while the automated analysis procedures
% in~\cite{li-tripunitara,stoller} only allows finitely many users and
% roles.  It is also more general than that of~\cite{stoller2} which is
% parametric in the roles but the number of users must be bounded.  On
% the other hand, it is less general than the decidability result
% in~\cite{stoller2} since it allows only a limited form of negation.
% On the other hand, our result is more general as in many situations we
% can eliminate the universal quantifiers needed to express $\neg
% ua^*(u,r)$ because the set of roles is known, say $R$, and $\neg
% ua^*(u,r)$ can be replaced with $\bigwedge_{r'\in R} (r'\succeq r
% \Rightarrow \neg ua(u,r'))$ as observed in
% Section~\ref{sec:symbolic-specification-ARBAC}.
% Notice that these results hold for every theory $T_{\mathit{ARBAC}}$
% specified as described in
% Section~\ref{sec:symbolic-specification-ARBAC}.  % ; hence, it subsumes the
% decidability results of~\cite{li-tripunitara,stoller,stoller2} and
% also for situations not considered before, such as the case when
% users are finitely many but their number is unknown.

Finally, notice that we can reduce other analysis problems (e.g., role
containment) to user-role reachability problems and thus show their
decidability.  For lack of space, this can be found in
%%% SR: COMMENT OR UNCOMMENT THE FOLLOWING LINES FOR THE SHORT OR LONG
%%% VERSION, RESP.
\cite{ext-version}.
% Appendix~\ref{app:related-problems}

%%% THE FOLLOWING SECTION HAS BEEN ADDED BECAUSE WE WERE GRANTED 1/2
%%% EXTRA PAGES!
\section{Preliminary experiments}
\label{sec:exp}

We briefly discuss some experiments with a prototype implementation of
the symbolic reachability procedure in Figure~\ref{fig:reach-algo}
that we call \fbk, short for Automated Symbolic Security Analyser.  We
consider the synthetic benchmarks described in~\cite{stoller} and
available on the web at~\cite{stollerpage} whereby both the number of
users and roles is bounded.  We perform a comparative analysis between
\fbk\ and the state-of-the-art tool in~\cite{stoller}, called
\stoller\ below.  Our findings shows that \fbk\ scales better than
\stoller\ on this set of benchmarks; the experiments were conducted on
an Intel(R) Core(TM)2 Duo CPU T5870, 2 GHz, 3 GB RAM, running Linux
Debian 2.6.32.

%\paragraph{Architecture of the system.}  
A client-server architecture is the most obvious choice to implement
the proposed symbolic backward reachability procedure.  The client
generates the sequence of formulae representing pre-images of the
formula representing the goal.  In addition, the client is also
assumed to generate the formulae characterising the tests for fix-point
or for non-empty intersection with the initial set of policies.  The
server performs the checks for satisfiability modulo
${T}_{\mathit{ARBAC}}$ and can be implemented by using
state-of-the-art automated deduction systems such as automated theorem
provers (in our case, SPASS~\cite{spass}) or SMT solvers (in our case,
Z3~\cite{z3}).  Although these tools are quite powerful, preliminary
experiments have shown that the formulae to be checked for
satisfiability generated by the client quickly become very large and
are not easily solved by available state-of-the-art tools.  A closer
look at the formulae reveals that they can be greatly simplified with
substantial speed-ups in the performances of the reasoning systems.
To this end, some heuristics have been implemented whose description
is not possible here for lack of space; the interested reader is
pointed to~\cite{asiaccs} for a complete description and more
experiments.
%%%%%%%%%%%%%%%%%%%%%%%%%%%%%%%%%%%%%%%%%%%%%%%%%%%%%%%%%%%%%%%%%%%%%%%%
\begin{figure}[t]
  \centering
  \begin{tabular}{cc}
    {\small goal size = 1} & {\small goal size = 2} \\
    \includegraphics[scale=.315]{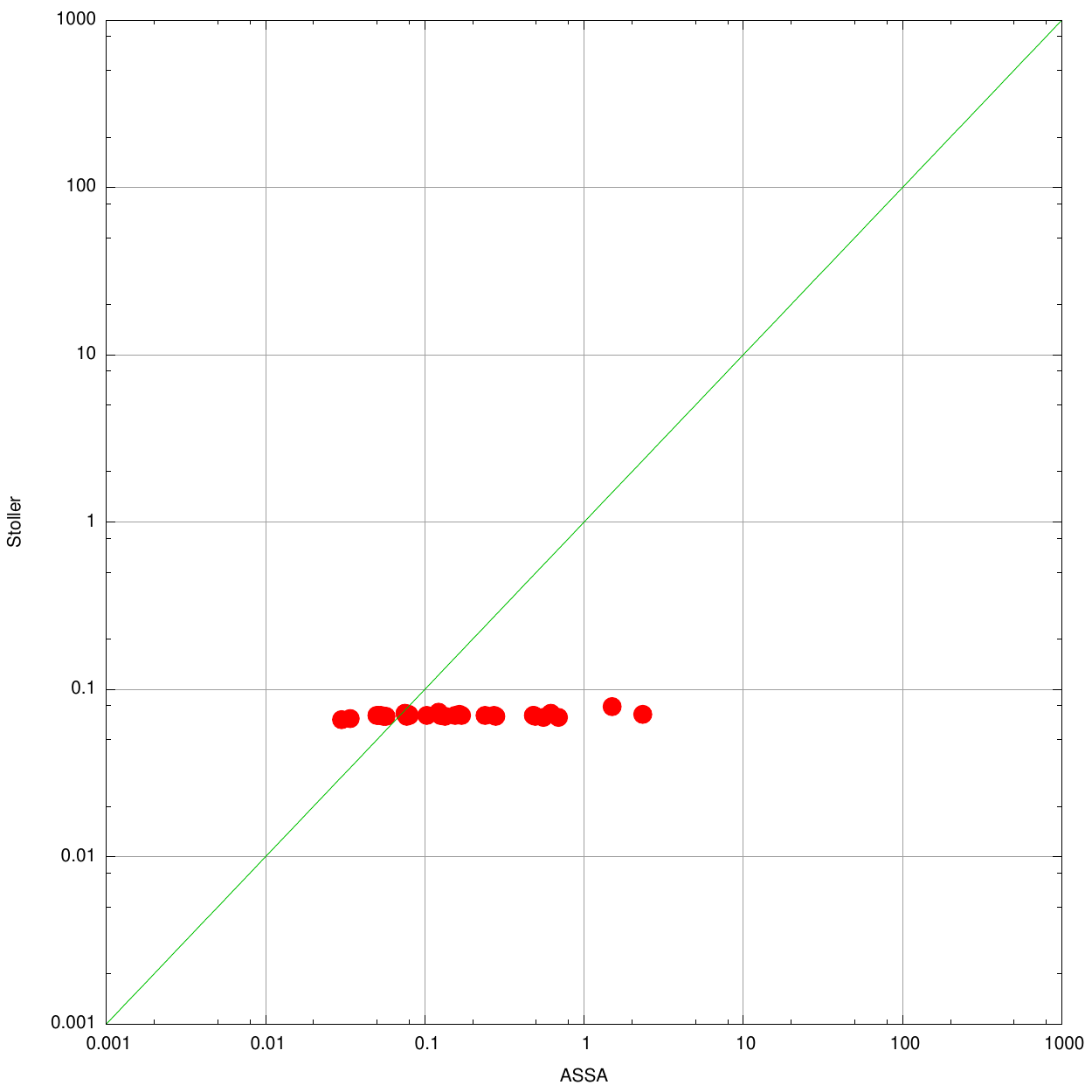} &
    \includegraphics[scale=.315]{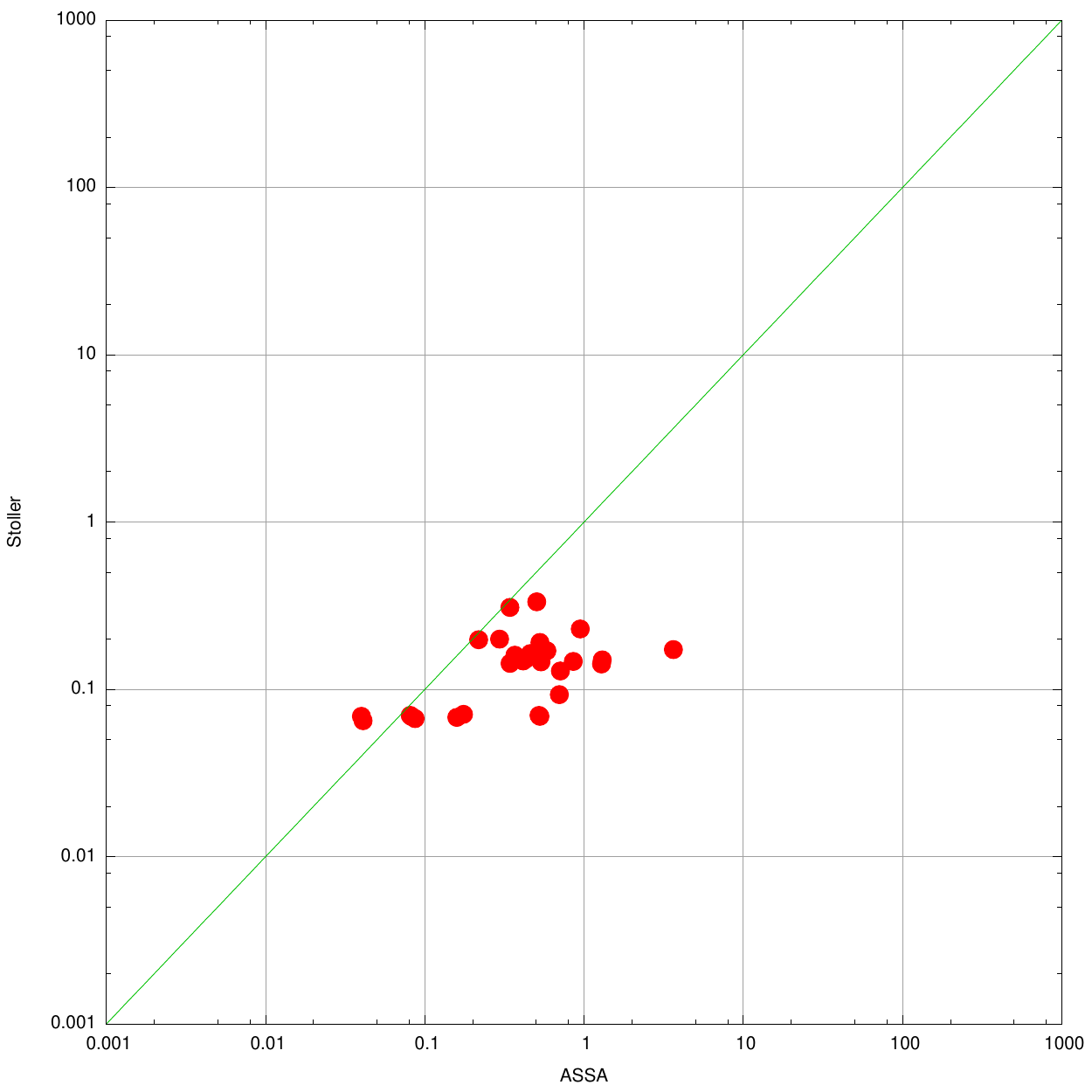} \\
    {\small goal size = 3} & {\small goal size = 4} \\
    \includegraphics[scale=.315]{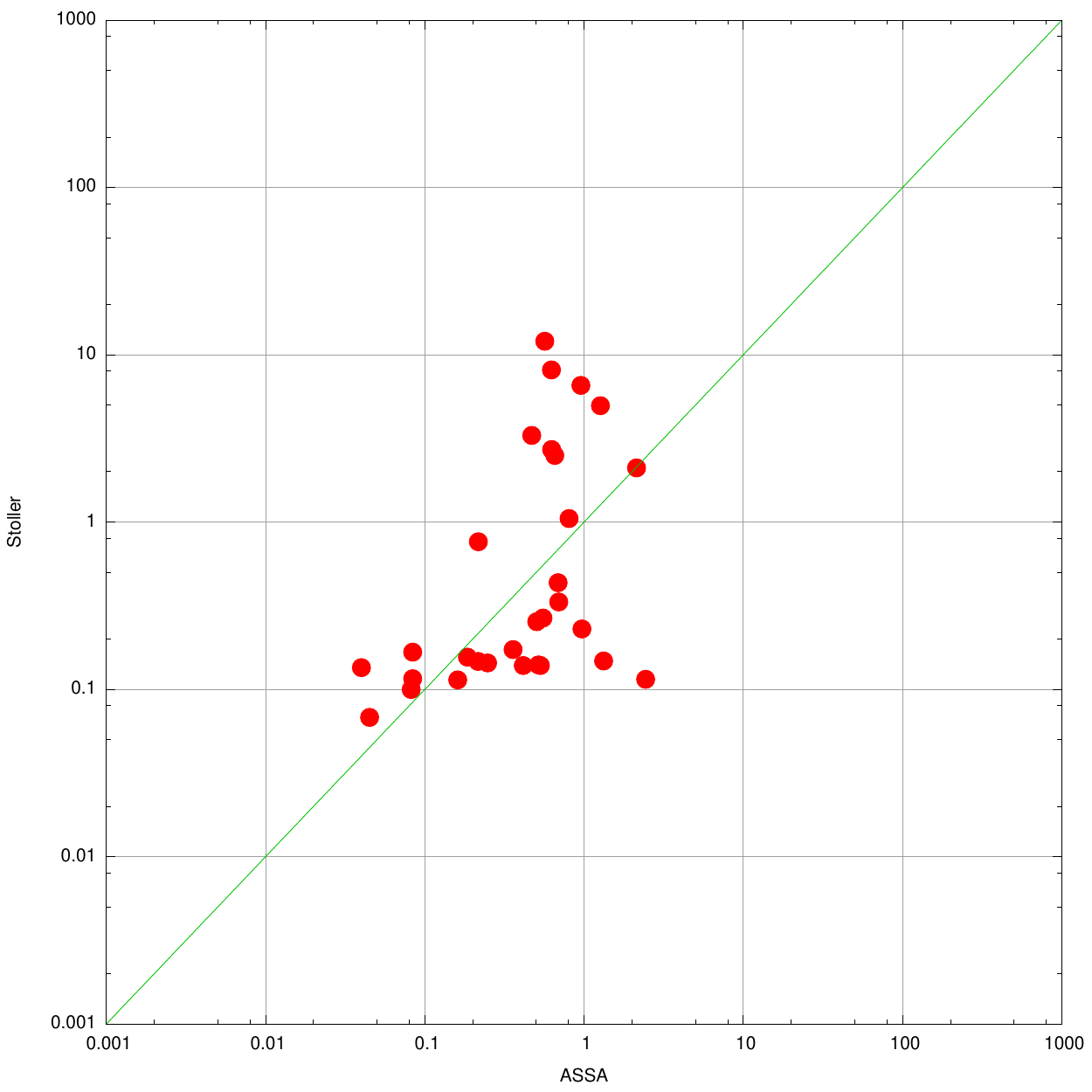} &
    \includegraphics[scale=.315]{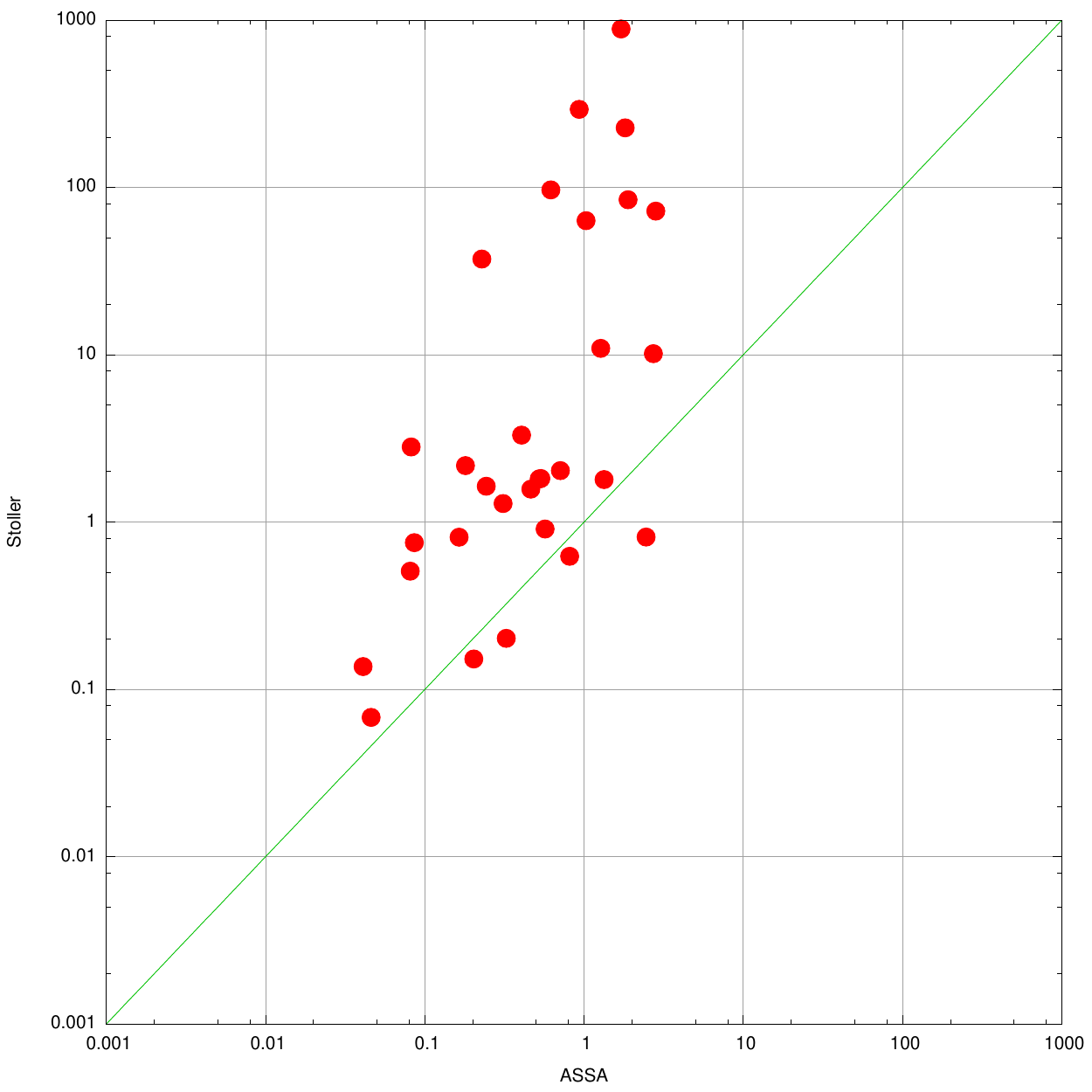} 
  \end{tabular}
  \caption{\label{fig:goal_size-vs-timings_numnodes_bench3}Comparison between \fbk\ and \stoller\ on some benchmarks from~\cite{stoller,stollerpage}}
\end{figure}
%%%%%%%%%%%%%%%%%%%%%%%%%%%%%%%%%%%%%%%%%%%%%%%%%%%%%%%%%%%%%%%%%%%%%%%%

%\paragraph{Benchmarks and results.}  
We consider the randomly generated benchmarks in~\cite{stollerpage},
where only the user-role assignment relation $ua$ can be modified by
\textsl{can\_assign} or \textsl{can\_revoke} actions (as assumed in
Section~\ref{sec:ARBAC-intro}).  These benchmarks were generated under
two additional simplifying assumptions: (i) a fixed number of users
and roles, and (ii) absence of role hierarchy (this is without loss of
generality under assumption (i) as observed in~\cite{stoller}).
% The problem generator allows one to choose the values of several
% parameters (e.g., the size of the goal); for a detailed description
% of the options the reader is pointed to~\cite{stoller,stollerpage}.
% The various classes were generated by selecting the values of the
% parameters according to those of an ARBAC policy of a university
% case study described in~\cite{stoller} and then varying selected
% parameters to explore the effects.
Besides the number of roles, one of the key parameter of the
benchmarks (according to the parametrised complexity result derived
in~\cite{stoller}) is the \emph{goal size}, i.e.\ the number of roles
in the set $RP$ of a goal reachability problem (as defined at the end
of Section~\ref{sec:ARBAC-intro}) or, equivalently, the number of
constants of sort $\mathit{Role}$ occurring in the symbolic goal
(\ref{eq:symbolic-goal}) of Section~\ref{subsec:reach-and-smt}.  The
benchmarks are divided in five classes.  The first and the second
classes were used to evaluate the worst-case behavior of forward
search algorithms (i.e.\ when the goal is unreachable) described
in~\cite{stoller}.  Our backward procedure (almost) immediately
detects unreachability by realizing that no action is backward
applicable.  The fourth and fifth classes of benchmarks fix the goal
size to one while the values of other parameters (e.g., the
cardinality of the set $R$ of roles) grow.  In particular, the fourth
class was used to show that the cost of analysis grows very slowly as
a function of the number of roles while the fifth aimed to compare the
performances of an enhanced version of the forward and the backward
algorithms of~\cite{stoller}.  For both classes, \fbk\ confirms that
its running time grows very slowly according to the results reported
in~\cite{stoller}.  However, \fbk\ is slightly slower than \stoller\
because of the overhead of invoking automated reasoning systems for
checking for fix-points instead of the \emph{ad hoc} techniques
of~\cite{stoller}.  The most interesting class of problems is the
third, which was used to evaluate the scalability of the backward
reachability algorithm of~\cite{stoller} with respect to increasing
values of the goal size $1,2,3,$ and
$4$. Figure~\ref{fig:goal_size-vs-timings_numnodes_bench3} shows four
scatter plots for values $1,2,3$, and $4$ of the goal size: the X and
Y axes report the median times of \fbk\ and \stoller, respectively
(logarithmic scale), to solve the $32$ reachability problems in the
third class of the benchmarks.  A dot above the diagonal means a
better performance of \fbk\ and viceversa; the time out was set to
$1,800$ sec.  Although, both \stoller\ and \fbk\ were able to solve
all the problems within the time-out, our tool is slower for goal
sizes $1$ and $2$, behaves as \stoller\ for goal size $3$, but
outperforms this for goal size $4$.  These results are encouraging and
seem to confirm the scalability of our techniques.  For a detailed
description of the implementation of \fbk\ and a more comprehensive
experimental evaluation (confirming these results), the reader is
pointed to~\cite{asiaccs}.

\section{Discussion}
\label{sec:conclusion}
We have proposed a symbolic framework for the automated analysis of
ARBAC policies that allowed us to prove the decidability of the
parametric reachability problem.  We used a decidable fragment of
first-order logic to represent the states and the actions of ARBAC
policies to design a symbolic procedure to explore the (possibly
infinite) state space.  Preliminary results with a prototype tool
implementing the backward reachability procedure in
Figure~\ref{fig:reach-algo} are encouraging.  A detailed description
of the implementation of the prototype and an extensive experimental
analysis is available in~\cite{asiaccs}.

% To the best of our knowledge, this is the first proposal of a general
% framework to study the parametric verification of ARBAC policies.
There are two main directions for future work.  First, it would be
interesting to study to what extent other variants of ARBAC can be
formalized in our framework, e.g., for UARBAC~\cite{uarbac}.
% is a schematic framework that allow arbitrary constraints over
% arbitrary types.  Another interesting extension where our symbolic
% approach can prove extremely useful is the analysis of RBAC with
% temporal constraints~\cite{tarbac} for role activation, role
% assignment, and so on.  Even more interesting would be to combine
% traditional ARBAC and its temporal variant: we believe this can be
% done in a very natural way in a suitable extension of the framework
% presented here.
Second, we want to adapt techniques developed in the context of
infinite state model checking to eliminate universal quantifiers in
guards of administrative actions (called global conditions, see,
e.g.,~\cite{abdulla-delzanno-rezine}), to allow for unrestricted
negation in \textsl{can\_assign}s.  

\paragraph{Acknowledgements.}  This work was partially supported by
the ``Automated Security Analysis of Identity and Access Management
Systems (SIAM)'' project funded by Provincia Autonoma di Trento in the
context of the ``team 2009 - Incoming'' COFUND action of the European
Commission (FP7), the FP7-ICT-2007-1 Project no.~216471, ``AVANTSSAR:
Automated Validation of Trust and Security of Service-oriented
Architectures,'' and the PRIN'07 Project 20079E5KM8 (Integrating
automated reasoning in model checking: towards push-button formal
verification of large-scale and infinite-state systems) funded by
MIUR.  Francesco Alberti must be thanked for his effort in
implementing and benchmarking \fbk.

%%%%%%%%%%%%%%%%%%%%%%%%%%%%%%%%%%%%%%%%%%%%%
\bibliographystyle{plain}
\bibliography{biblio}
%%%%%%%%%%%%%%%%%%%%%%%%%%%%%%%%%%%%%%%%%%%%%

%%% SR: COMMENT OR UNCOMMENT THE FOLLOWING LINE TO HAVE SHORT OR LONG
%%%     VERSION, RESP.
\newpage

\appendix

\section*{Plan of the Appendixes}

\noindent We provide some additional material to illustrate and
integrate the results presented in the paper:
\begin{itemize}
\item Appendix~\ref{app:extension-parametric-role} discusses how to
  formalize parametric roles in our framework and explain that the
  decidability result for user-role reachability also cover this scenario.
\item Appendix~\ref{app:termination} presents the formal details of the
  termination of the backward reachability procedure in
  Figure~\ref{fig:reach-algo}.  
\item Appendix~\ref{app:related-problems} discusses three related
  security analysis problems for ARBAC polices (namely, inductive
  policy invariant, role containment, and weakest preconditions) and
  their relationship with the user-role reachability problem.
\item Finally, Appendix~\ref{app:br-illustration} describes in some
  detail an execution of the symbolic backward reachability procedure
  Figure~\ref{fig:reach-algo} on a simple example taken
  from~\cite{stoller}.
\end{itemize}

\section{Formalizing parametric roles}
\label{app:extension-parametric-role}

Here, we explain how it is possible to model ARBAC policies with
parametrised roles as considered in, e.g.,~\cite{stoller2}.

A role schema can be seen as an expression of the form $\rho(p_1, ...,
p_n)$ for $n\geq 0$, where $\rho$ is a role name and $p_i$ is a
distinct parameter name $i=1, ..., n$.  Each parameter can take values
from a given data type containing an infinite number of values.  An
instance of a role schema is an expression of the form $\rho(p_1=t_1,
..., p_n=t_n)$, where $t_i$ is a data value or a variable.  For
example, in the university policy considered in~\cite{stoller2}, the
role schema $\mathit{Student}(dept,cid)$ is used for students
registered for the course numbered $cid$ offered by department $dept$,
the role schema $\mathit{Student}(dept)$ is used for all students of a
specific department $dept$, and the instance
$\mathit{Student}(dept=cs,cid=101)$ identifies students of the
Computer Science department taking course $101$.  Role schemas can be
overloaded by using parameter names; e.g., $\mathit{Student}$ can have
one parameter named $dept$ or two parameters named $dept$ and $cid$.
A parametrised version of ARBAC policies can use parametric roles to
express role assignment and revocation in a very compact way.  For
example, in the case of the university policy, one can have the
following role schemas: $\mathit{Chair}(dept)$,
$\mathit{Student}(dept, cid)$, and $\mathit{TA}(dept,cid)$.  Then, a
\textsl{can\_assign} rule is the following: the chair of department $D$ (i.e.\
a user belonging to the role $\mathit{Chair}(dept=D)$) can assign a
student of a department $D$ taking course $cs$ (i.e.\ a user belonging
to the role $\mathit{Student}(dept=D, cid=CID)$) to be the teaching
assistant of that course (i.e.\ a user belonging to the role
$\mathit{TA}(dept=D, cid=CID)$).

In our symbolic framework, this situation can be formalized as
follows.  We introduce a predicate symbol extended with an extra
argument for each parametric role, i.e.\ if the number of role names
in the role schema $\rho$ is $n$, then we use a predicate symbol
$\rho$ of arity $n+1$ (this technique is standard for example to
translate Entity-Relationship diagram schemas to fragments of
first-order logic).  For the example above, we introduce the following
predicate symbols: $\mathit{Chair}$, $\mathit{Student}$, and
$\mathit{TA}$ of arity 2, 3, and 3, respectively.  We do not use
parameter names, instead we fix an order on them so that we can use
the standard way of building atoms in first-order logic.  When a role
schema is overloaded, we introduce a different predicate symbol in
order to disambiguate the situation; a simple automated pre-processing
phase can be used to eliminate overloading.  In this context, the `can
assign rule above can be written as follows:
\begin{eqnarray*}
  \exists u,r,D,CID,u_1,r_1,r_2.\left(
    \begin{array}{ll}
      \mathit{Chair}(D,r) \wedge ua(u,r)  & \wedge  \\
      \mathit{Student}(D, CID, r_1) \wedge ua(u_1,r_1) & \wedge  \\
      \mathit{TA}(D, CID, r_2) \wedge  ua' = ua \oplus (u_1,r_2) &
      \end{array} 
    \right) ,
\end{eqnarray*}
where the variables $r, r_1,$ and $r_2$ are used as the names of the
roles corresponding to the particular value of the attributes in the
role schema.  This means that we need to require that each relation is
functional or, equivalently, that the interpretation of the predicate
symbols must be partial functions.  In our framework, this can be done
by adding suitable formulae to the background theory
$T_{\mathit{ARBAC}}$.  For the example of the university policy
considered above, we can simply write the following two
$\forall$-formulae:
\begin{eqnarray*}
  \forall D,r_1,r_2.((\mathit{Chair}(D,r_1)\wedge \mathit{Chair}(D,r_2)) 
   & \Rightarrow & r_1 = r_2) \\
  \forall D,CID,r_1,r_2.((\mathit{Student}(D,CID,r_1)\wedge 
                          \mathit{Student}(D,CID,r_2)) 
   & \Rightarrow & r_1 = r_2) \\
  \forall D,CID,r_1,r_2.((\mathit{TA}(D,CID,r_1)\wedge 
                          \mathit{TA}(D,CID,r_2)) 
   & \Rightarrow & r_1 = r_2) .
\end{eqnarray*}
Notice also that we can specify additional constraints among two or
more relations if we can express them as $\forall$-formulae.  It is
not obvious how this feature can be added to the approach
in~\cite{stoller2}.  For the example above, we have mentioned that we
can have a role schema $\mathit{Student}(dept)$ for identifying all
students in the department $dept$.  Indeed,
$\mathit{Student}(dept,cid)$ must characterize sub-sets of users of
the role $\mathit{Student}(dept)$.  If we introduce a binary predicate
symbol $\mathit{Student}_1$ of arity 2 corresponding to the role
schema $\mathit{Student}(dept)$, then we can express this by the
following $\forall$-formula:
\begin{eqnarray*}
  \forall D,CID,r.(\mathit{Student}(D,CID,r) \Rightarrow
                   \mathit{Student}_1(D,r) ) ,
\end{eqnarray*}
which can be added to $T_{\mathit{ARBAC}}$.

To summarize, our framework can handle parametrised roles as follows.
First, the sub-theory $T_{Role}$ of $T_{ARBAC}$ becomes many-sorted:
besides the sort $Role$, we introduce as many sort symbols---called
\emph{parameter sorts}---as domains for the parameters of each role.
Furthermore, for each role symbol $\rho$ of arity $n$, we introduce a
predicate symbol of arity $n+1$.  Overloading is eliminated by
introducing decorated versions of the predicate symbol and an order on
the parameter names is fixed so that we can use the standard way of
building atoms of first-order logic.  Second, for each predicate
symbol $\rho$ of arity $n+1$, we add the following functional
constraint to $T_{\mathit{Role}}$ and hence to $T_{ARBAC}$:
\begin{eqnarray*}
  \forall \underline{x}, r_1,r_2.((\rho(\underline{x},r_1) \wedge \rho(\underline{x},r_2)) \Rightarrow r_1 = r_2) ,
\end{eqnarray*}
where $\underline{x}$ is a tuple of length $n$ of variables of
appropriate sorts.  If needed, we can add further constraints, (e.g.,
formalizing relationship between different role symbols) if these can
be expressed as $\forall$-formulae.  For example, it is worth noticing
how to express the role hierarchy for parametrised role.  Besides the
usual axioms requiring $\succeq$ to be a partial order, we can add
also $\forall$-formulae of the following form:
\begin{eqnarray*}
  \forall \underline{x}, \underline{y}, r_1,r_2.((\rho_1(\underline{x},r_1) \wedge \rho_2(\underline{y},r_2)) \Rightarrow r_1 \succeq r_2) ,
\end{eqnarray*}
where $\underline{x},\underline{y}$ are tuples of variables of
appropriate sorts, $\rho_1,\rho_2$ are two predicates representing
parametric roles.  This axiom requires that all instances of the
parametric role $\rho_1$ are senior than those of role $\rho_2$.
Notice that one can design more sophisticated hierarchical
relationships between role instances depending on the values of the
parameters, provided that the signature is rich enough to express the
constraints between the values of the parameters and that only
$\forall$-formulae are used.

% Concerning cardinality on the domains of the parameter names, for
% example, we can require that certain domains have a fixed and finite
% cardinality (say, e.g., $2$) by means of the following two
% $\forall$-formulae:
%   \begin{eqnarray*}
%     \forall x.(x=n_1 \vee x=n_2) & \mbox{ and } &
%     n_1\neq n_2
%   \end{eqnarray*}
%   where $n_1$ and $n_2$ are constant symbols and $x$ ranges over the
%   type symbol interpreted as the domain of the parameter under
%   consideration.

Finally, \textsl{can\_assign} and \textsl{can\_revoke} actions can be written by using
existentially quantified variables ranging over the parameter names
besides those ranging over users and roles; thus generalizing the
shapes of actions (\ref{eq:can_assign}) and (\ref{eq:can_revoke}).
Formally, transitions have the following forms:
\begin{eqnarray*}
  % \label{eq:can_assign}
  \exists u,r, u_1, r_1, r_2, ..., r_k, \underline{p}. & 
  (C(u,r, u_1, r_1, r_2, ..., r_k, \underline{p})  \wedge  & ua' = ua \oplus (u_1,e^r)) \\
  % \label{eq:can_revoke}
  \exists u,r,u_1,\underline{p}.&(C(u,r, u_1,\underline{p}) ~~~~~~~~~~~~~~~~\wedge  & ua' = ua \ominus (u_1,e^r))
\end{eqnarray*}
where $\underline{p}$ is a tuple of variables of parameter sorts and
$C$ is a constraint in which also literals built out of the predicate
symbols introduced for modelling parametric roles may occur.

All the results proved in Sections~\ref{sec:reachability-analysis}
and~\ref{sec:reachability-analysis} can be easily extended to cover
ARBAC policies with parametric roles as soon as we observe that the
formulae introduced here satisfy the assumptions on the theory
$T_{\mathit{ARBAC}}$ of
Section~\ref{sec:symbolic-specification-ARBAC}.

\section{Termination of backward reachability}
\label{app:termination}

\subsection{Pre- and well-quasi-orders: definitions and basic properties}
\label{subapp:pre-and-well-quasi-orders}

A \emph{pre-order} $(P,\leq)$ is the set $P$ endowed with a reflexive
and transitive relation.  We say that $\leq$ is \emph{decidable} if,
given $p_1$ and $p_2$ in $P$, we can algorithmically check whether
$p_1\leq p_2$.  An \emph{upward closed set $U$} of the pre-order
$(P,\leq)$ is such that $U\subseteq P$ and if $p\in U$ and $p\leq q$
then $q\in U$.  A \emph{cone} is an upward closed set of the form
$\uparrow p = \{ q\in P ~|~ p\leq q\}$. An upward closed set $U$ is
\emph{finitely generated} iff it is a finite union of cones.  

For an upward closed set $U$, a \emph{generator} of $U$ is a set $G$
such that (a) $U=\bigcup_{g\in G} \uparrow g$ and (b) $g_1\leq g_2$
implies $g_1=g_2$, for every $g_1,g_2\in G$.  It is easy to see that
$G$ contains only minimal elements (w.r.t.\ $\leq$) but, in general,
it needs not to be unique.  In any case, it is always possible to
define a function $gen(U)$ returning a unique generator of $U$ (the
same chosen among the many possible ones).  

A pre-order $(P,\leq)$ is a \emph{well-quasi-ordering (wqo)} iff every
upward closed sets of $P$ is finitely generated (this is equivalent to
the standard definition of wqo, see~\cite{ijcar08} for a proof).  In
the case of a wqo, $gen(U)$ is finite because of property (b) of the
definition of generator of $U$.  This implies that every upward closed
set $U$ can be characterized by a finite set of configurations, namely
$gen(U)$.  

\subsection{Some notions and results of model-theory}
\label{subapp:model-theory}

Let $\mathcal{M}$ be a $\Sigma$-structure. A \emph{substructure} of
$\mathcal{M}$ is a $\Sigma$-structure $\mathcal{N}$ whose domain is
contained in that of $\mathcal{M}$ and such that the interpretations
of the symbols of $\Sigma$ in $\mathcal{N}$ are restrictions of the
interpretation of these symbols in $\mathcal{M}$; conversely, we say
that $\mathcal{M}$ is a \emph{superstructure} of $\mathcal{N}$.  Let
$\mathcal{C}$ be a class of structures; we say that $\mathcal{C}$ is
\emph{closed under substructures} if $\mathcal{M}\in \mathcal{C}$ and
$\mathcal{N}$ is a substructure of $\mathcal{M}$, then $\mathcal{N}\in
\mathcal{C}$. 
\begin{property}
  A class $\mathcal{C}$ of structures is closed under substructures
  iff there exists a theory $T$ such that $T$ contains only
  $\forall$-formulae and $Mod(T)=\mathcal{C}$.
\end{property}
A proof of this result can be found in any book on model theory,
e.g.,~\cite{hodges}.

Let $\mathcal{M}$ and $\mathcal{N}$ two structures over the same
signature $\Sigma$ and $M,N$ be their domains, respectively; an
\emph{embedding} $s$ is an injective mapping from $M$ to $N$ such that
(i) $s(f^{\mathcal{M}}(e_1, ...., e_n)) = f^{\mathcal{N}}(s(e_1), ...,
s(e_n))$ for each function symbol $f$ in the signature $\Sigma$ and
(ii) $(e_1, ..., e_n)\in R^{\mathcal{M}}$ iff $(s(e_1), ...,
s(e_m))\in {\mathcal{N}}$ for each predicate symbol $R$ in $\Sigma$,
where $(e_1, ..., e_n)$ is a tuple of elements in $M$ of length equal
to the arity of $f$ or $R$, respectively.  In other words, an
embedding is a homomorphism that preserves and reflects relations.  It
is possible to show (see, e.g.,~\cite{hodges}) that any embedding can
be seen as the composition of an isomorphism followed by an
``extension,'' i.e.\ if there is an embedding from $\mathcal{M}$ to
$\mathcal{N}$, we can assume that $\mathcal{M}$ is a substructure of
$\mathcal{N}$ (or dually, $\mathcal{N}$ is a superstructure of
$\mathcal{M}$).  

Abstractly, (Robinson) diagrams give a logical formulation of model
theoretic properties such as ``there exists an embedding from
structure $\mathcal{M}$ to structure $\mathcal{N}$.''  The importance
of this will be clear when considering the definition of the pre-order
on configurations (given in terms of the existence of an embedding
between structures).  Let $\mathcal{M}$ be a $\Sigma$-structure and
$A$ be a sub-set of the domain of $\mathcal{M}$; $\Sigma(A)$ is the
signature obtained by adding to $\Sigma$ new symbols of constants $a$
for $a\in A$.  We can regard $\mathcal{M}$ as a $\Sigma(A)$-structure
when the interpretation function of $\mathcal{M}$ is extended so that
every element $a$ in $A$ is mapped to the constant $a$.  The
\emph{(Robinson) diagram of $A$ in $\mathcal{M}$}, in symbols
$\delta_{\mathcal{M}}(A)$, is the set $L$ of all $\Sigma(A)$-literals
such that $\mathcal{M}\models \ell$, for every $\ell\in L$.
\begin{lemma}[Diagram Lemma]
  \label{lem:diagram}
  Let $\mathcal{M}$ and $\mathcal{N}$ be two $\Sigma$-structures and
  $M$ be the domain of $\mathcal{M}$. Then, there exists an embedding
  from $\mathcal{M}$ to $\mathcal{N}$ iff $\mathcal{N}$ can be
  expanded to a $\Sigma(M)$-structure which is a model of
  $\delta_{\mathcal{M}}(M)$.
\end{lemma}
The proof of this fact is an immediate consequence of the definition
of Robinson diagram given above and can be found in any book on model
theory (see, e.g.,~\cite{hodges}).

\subsection{A pre-order on configurations: formal definition}

Let $\Gamma$ be a symbolic ARBAC policy, i.e.\
\begin{eqnarray*}
  \Gamma & := & (In(ua),
  \{\tau_1(ua,ua'), ..., \tau_n(ua,ua')\}, 
  \{\iota_1(ua), ..., \iota_m(ua)\} )
\end{eqnarray*}
where $In$ is a $\forall$-formula, $\iota_j$ is a $\forall$-formula,
and $\tau_i$ is a transition formula of the forms
(\ref{eq:can_assign}) and (\ref{eq:can_revoke}).  

Recall that a {state} of the ARBAC policy $\Gamma$ is a structure
$\mathcal{M}\in Mod(T_{\mathit{ARBAC}})$.
\begin{definition}
  A \emph{configuration} of $\Gamma$ is a state $\mathcal{M}$ where
  $\mathcal{M}$ is a finite model, i.e.\ the cardinality of the domain
  of $\mathcal{M}$ is bounded.
\end{definition}
We are now in the position to define the pre-order on configurations.
\begin{definition}
  Let $\mathcal{M}$ and $\mathcal{M}'$ be configurations.  We write
  $\mathcal{M} \leq \mathcal{M}'$ iff there exists an embedding $s$
  from $\mathcal{M}$ to $\mathcal{M}'$.  % If $\sigma \leq \sigma'$ and
%   the embedding is an inclusion mapping, then we say that $\sigma'$ is
%   a \emph{sub-configuration} of $\sigma$.
\end{definition}

\subsection{From $\exists$-formulae to configurations...}

We show that $\exists$-formulae identify configurations.  To state
this result formally, we recall the following notation: $[[K]] := \{
\mathcal{M}\in Mod(T_{\mathit{ARBAC}}) ~|~ \mathcal{M} \models K \}$,
for $K$ an $\exists$-formula.
\begin{proposition}
  \label{prop:exists-formulae-denote-upward-closed-sets}
  For every $\exists$-formula $K$, the set $[[K]]$ is upward closed.
\end{proposition}
\begin{proof}
  Since the union of an upward closed set is still an upward closed
  set, we assume---without loss of generality---that $K(ua)$ is of the
  form $\exists \underline{r}, \underline{u}.\varphi(\underline{r},
  \underline{u}, ua)$ where $\underline{u},\underline{r}$ are tuples
  of variables for users and roles, respectively, and $\varphi$ is a
  conjunction of literals (as we can always transform a Boolean
  combination of atoms into disjunctive normal form and then
  distribute the existential quantifiers over the disjunction).  Under
  these assumptions, showing that $[[K]]$ is upward closed amounts to
  prove that if the configuration $\mathcal{M}\in [[K]]$ and
  $\mathcal{M} \leq \mathcal{N}$, then $\mathcal{N}\in [[K]]$, i.e.\
  $\mathcal{N} \models K$.  Now, assume that $\mathcal{M}\in [[K]]$
  and $\mathcal{M} \leq \mathcal{N}$.  This implies, by definition of
  $[[\cdot]]$, that (a) $\mathcal{M} \models K$ and (b) there exists
  an embedding $s$ from $\mathcal{M}$ to $\mathcal{N}$.  From (a), by
  definition of truth, it follows that there exist tuples
  $\underline{e^u}$ and $\underline{e^r}$ of sort $\mathit{User}$ and
  $\mathit{Role}$, respectively, such that $\mathcal{M} \models
  K(\underline{e^u}, \underline{e^r})$.  From (b) and the definition
  of embedding, we derive that
  \begin{eqnarray*}
    \mathcal{M} \models K(\underline{e^u}, \underline{e^r})
    & \mbox{ iff } &
    \mathcal{N} \models K(s(\underline{e^u},\underline{e^r})) .
  \end{eqnarray*}
  The last two facts (and the well-known property that truth of
  quantifier-free formulae is preserved when considering
  superstructures) imply that $\mathcal{N} \models K$, as
  desired.  This concludes the proof that $[[K]]$ is upward
  closed. \qed
\end{proof}
We show that entailment between $\exists$-formulae is equivalent to
containment among configurations.
\begin{proposition}
  \label{prop:implication-exists-formulae-equal-containement-upward-closed-sets}
  $[[K_1]]\subseteq [[K_2]]$ iff $K_1\Rightarrow K_2$ is valid modulo
  $T_{\mathit{ARBAC}}$, for every pair of $\exists$-formulae
  $K_1,K_2$,
\end{proposition}
\begin{proof}
  There are two cases to consider.  The `if' case is trivial: it is an
  immediate consequence of the definition of truth.  For the `only if'
  case, we prove that if $K_1\Rightarrow K_2$ is not valid modulo
  $T_{\mathit{ARBAC}}$, then $[[K_1]]\not \subseteq [[K_2]]$.
  Assuming that $K_1\Rightarrow K_2$ is not valid modulo
  $T_{\mathit{ARBAC}}$ is equivalent, by refutation, to say that $\neg
  (K_1\Rightarrow K_2)$ (or $K_1\wedge \neg K_2$) is satisfiable
  modulo $T_{\mathit{ARBAC}}$.  In turn, this implies that $K_1\wedge
  \neg K_2$ is satisfiable in a finite model according to the proof of
  the decidability of the BSR class.  From this and
  Proposition~\ref{prop:exists-formulae-denote-upward-closed-sets}, we
  can derive that $[[K_1]]\cap [[K_2]]^c \neq \emptyset$ (where
  $\cdot^c$ denotes the set complement operation).  By simple
  set-theoretic manipulations, we derive $[[K_1]]\not\subseteq
  [[K_2]]$, as desired. \qed
\end{proof}
Lemma~\ref{lem:semantic-counterpart} is an immediate consequence of
Propositions~\ref{prop:exists-formulae-denote-upward-closed-sets}
and~\ref{prop:implication-exists-formulae-equal-containement-upward-closed-sets}.

\subsection{... and viceversa: from configurations to $\exists$-formulae}

We show that finitely generated upward closed sets of configurations
are configurations of the form $[[K]]$, for some $\exists$-formula
$K$.  To do this, we use Robinson diagrams (introduced in
Section~\ref{subapp:model-theory}) since they give a logical
formulation of model theoretic properties such as ``there exists an
embedding from structure $\mathcal{M}$ to structure $\mathcal{N}$.''
The importance of this is clear as soon as we recall the definition of
pre-order over configurations that requires the existence of an
embedding among structures to show that a configuration precedes
another according to the pre-order.  The main obstacle in using
diagrams is that the formula $\delta_{\mathcal{M}}(M)$ usually
contains infinitely many literals.  Fortunately, in our case, it is
possible to show that we can consider only a finite sub-set of
literals in $\delta_{\mathcal{M}}(M)$ as all the others are implied by
those in the sub-set.
\begin{proposition}
  \label{prop:finitely-generated-configs-equiv-exists-formulae}
  The following facts hold:
  \begin{description}
  \item[(i)] with every configuration $\mathcal{M}$, it is possible to
    effectively associate an $\exists$-formula $K_{\mathcal{M}}$
    (called \emph{diagram formula (for $\mathcal{M}$)}) such that
    $[[K_{\mathcal{M}}]] = \uparrow \mathcal{M}$,
  \item[(ii)] with every $\exists$-formula $K$, it is possible to
    effectively associate a finite set $\{ \mathcal{M}_1, ...,
    \mathcal{M}_n\}$ of configurations such that $K$ is equivalent to
    $\bigvee_{i=1}^n K_{\mathcal{M}_i}$,
  \item[(iii)] any finitely generated upward closed set of
    configurations coincides with $[[K]]$, for some $\exists$-formula
    $K$.
  \end{description}
\end{proposition}
\begin{proof}
  We consider the three cases separately.
  \begin{description}
  \item[(i)] Let $\mathcal{M}$ be a configuration and consider the
    ``diagram'' $\delta_{\mathcal{M}}(\underline{e^u},
    \underline{e^r})$ where $\underline{e^u}$ and $\underline{e^r}$
    are finite tuples of users, roles, and permissions, respectively,
    that are also in the domain of $\mathcal{M}$.
    \begin{remark}
      Notice that $\delta_{\mathcal{M}}(\underline{e^u},
      \underline{e^r})$ is not the Robinson diagram as defined above;
      however, it turns out to be equivalent to
      $\delta_{\mathcal{M}}(\{e^u_i ~|~ i\geq 0\} \cup \{e^r_i ~|~
      i\geq 0\})$, i.e.\ the ``real'' diagram.  This is so because in
      any model of a BSR theory, there are only finitely many distinct
      atoms that ``matter,'' which are precisely those in
      $\delta_{\mathcal{M}}(\underline{e^u}, \underline{e^r})$,
      because when checking for satisfiability we can always restrict
      to those constants that occur in the formula to be checked for
      satisfiability as discussed in the sketch of the proof of
      Property~\ref{lem:dec-URA}.  (Recall, in fact, that by applying
      Herbrand theorem, the Herbrand universe is finite and composed
      only of the constants occurring in the formula.)  So, below, we
      refer to $\delta_{\mathcal{M}}(\underline{e^u},
      \underline{e^r})$ as the diagram and we treat it as the
      conjunction of its elements (i.e.\ as a first-order formula)
      since it is finite. \qed
    \end{remark}
    Now, take $K_{\mathcal{M}}$ to be the following $\exists$-formula:
    \begin{math}
      \exists  \underline{u}, \underline{r}. 
        \delta_{\mathcal{M}}(\underline{u}, \underline{r}) .
    \end{math}
    We are left with the problem of proving that $[[K_{\mathcal{M}}]]
    = \uparrow \mathcal{M}$.  By the definitions of
    $[[K_{\mathcal{M}}]]$ and $\uparrow \mathcal{M}$, this is
    equivalent to show that a configuration $\mathcal{N}$ is in
    $[[K_{\mathcal{M}}]]$, or---equivalently---$\mathcal{N} \models
    \exists \underline{u},
    \underline{r}. \delta_{\mathcal{M}}(\underline{u}, \underline{r})$
    iff $\mathcal{M}\leq \mathcal{N}$.  Now, assume $\mathcal{N}
    \models \exists \underline{u},
    \underline{r}. \delta_{\mathcal{M}}(\underline{u},
    \underline{r})$, which is equivalent to $\mathcal{N} \models
    \delta_{\mathcal{M}}(\underline{e^u}, \underline{e^r})$.  By the
    Diagram Lemma (i.e.\ Lemma~\ref{lem:diagram} above), this is
    equivalent to the existence of an embedding from $\mathcal{M}$ to
    $\mathcal{N}$, which---in turn---is equivalent to $\mathcal{M}\leq
    \mathcal{N}$, by definition of $\leq$.

  \item[(ii)] Without loss of generality, we can assume $K$ to be
    $\exists \underline{u},\underline{r}.\bigvee_{k=1}^n
    \varphi_k(\underline{u},\underline{r})$.  For each $k=1, ..., n$,
    we can also assume (again without loss of generality) that there
    exists an existentially quantified variable $x$ in
    $\underline{u}\cup \underline{r}$ such that $x=t$, for each
    constant in $K$. In this way, all the elements are explicitly
    mentioned in $K$.  Now, in a BSR theory, every quantifier-free
    formula with at most $m$ free variables is equivalent to a
    disjunction of the diagram $\delta_{\mathcal{M}}(X)$ where
    $\mathcal{M}$ is a substructure of a model in the theory and $X$
    is a set of elements of cardinality at most $m$.  Thus, $K$ can be
    rewritten as
    \begin{eqnarray*}
      \bigvee_{\mathcal{A}} \exists \underline{u},\underline{r}.
        \delta_{\mathcal{A}}(\underline{u},\underline{r})
    \end{eqnarray*}
    for $\mathcal{A}$ ranging over the models whose cardinality is $m$
    (recall that the class of models of a BSR theory is closed under
    substructures).  Each disjunct can be unsatisfiable, because it
    does not agree with the interpretation of $ua$, or satisfiable
    and, in this case, the model $\mathcal{A}$ is a configuration such
    that $\exists \underline{u},\underline{r}.
    \delta_{\mathcal{A}}(\underline{u},\underline{r})$ is precisely
    $K_{\mathcal{A}}$, as desired.
  \item[(iii)] An immediate corollary of (i) and (ii) above. \qed
  \end{description}
\end{proof}
The results in this and the previous subsection tells us that
$\exists$-formulae and configurations can be used interchangeably.

\subsection{Proof of termination of backward reachability}

\noindent \textbf{Theorem~\ref{thm:termination}.}  The backward
reachability procedure in Figure~\ref{fig:reach-algo} terminates.
\begin{proof}
  First of all, notice that when the algorithm return
  $\mathsf{reachable}$, it also terminates (line 3).  So, we consider
  the case when the goal is unreachable.  Let
  $B(\tau,K):=\bigcup_{i\geq 0} [[BR^i(\tau, K)]]$.  There two cases
  to consider.
  \begin{itemize}
  \item Let $K$ be the $\exists$-formula given in input to the
    algorithm and assume that $B(\tau,K)$ is finitely generated (that
    $B(\tau,K)$ is an upward closed set is obvious because it is
    obtained as union of upward closed sets since $[[K]]$ is so by
    Proposition~\ref{prop:exists-formulae-denote-upward-closed-sets}).
    Because of
    Proposition~\ref{prop:implication-exists-formulae-equal-containement-upward-closed-sets},
    we have that
    \begin{eqnarray*}
      [[BR^0(\tau, K)]] \subseteq [[BR^2(\tau, K)]] \subseteq \cdots
      \subseteq  [[BR^n(\tau, K)]]  \subseteq [[BR^{n+1}(\tau, K)]]
      \subseteq \cdots
    \end{eqnarray*}
    Because $B(\tau,K)$ is finitely generated, we have that there
    exists $n$ such that $[[BR^n(\tau, K)]] = [[BR^{n+1}(\tau, K)]]$
    and, again by
    Proposition~\ref{prop:implication-exists-formulae-equal-containement-upward-closed-sets},
    we derive that $BR^n(\tau, K) \Leftrightarrow BR^{n+1}(\tau, K)$
    is valid modulo $T_{\mathit{ARBAC}}$, i.e.\ the algorithm halts.
  \item Assume that the algorithm terminates.  By
    Proposition~\ref{prop:implication-exists-formulae-equal-containement-upward-closed-sets},
    this is equivalent to $BR^n(\tau, K) \Leftrightarrow
    BR^{n+1}(\tau, K)$ is valid modulo $T_{\mathit{ARBAC}}$ which, by
    Proposition~\ref{prop:implication-exists-formulae-equal-containement-upward-closed-sets},
    is equivalent to $[[BR^n(\tau, K)]] = [[BR^{n+1}(\tau, K)]]$, for
    some $n\geq 0$.  Notice that $B(\tau,K)=[[BR^n(\tau, K)]]$ is
    finitely generated by
    Proposition~\ref{prop:finitely-generated-configs-equiv-exists-formulae}.
  \end{itemize}
  So far, we have proved that the backward reachability procedure in
  Figure~\ref{fig:reach-algo} terminates iff $B(\tau,K)$ is finitely
  generated.  Thus, to conclude the proof, we show that $B(\tau,K)$ is
  indeed finitely generated.  To this end, if we are able to prove
  that the pre-order on configurations is a wqo, then are entitled to
  conclude that $B(\tau,K)$ is finitely generated (recall the
  definition of wqo in
  Section~\ref{subapp:pre-and-well-quasi-orders}).  Now, the pre-order
  on configurations is a wqo by Dickson's Lemma~\cite{dickson}. In
  fact, a configuration is uniquely determined by a pair of integers
  counting the number of pairs $(u,r)$ for which $ua(u,r)$ holds and
  the configuration ordering is obtained by component-wise comparison.
  This concludes the proof.  \qed
\end{proof}
Combining the results above, we derive the main result of this paper,
i.e.\ Theorem~\ref{thm:goal-reach-dec}.
% \begin{theorem}
%   The goal reachability problem for ARBAC policies with a finite but
%   unknown number of roles and users is decidable.
%  \end{theorem}
% Concerning complexity, the situation is less pleasant than the nice
% results reported in~\cite{li-tripunitara,stoller}.  In fact, it is
% possible to show that there is a non-primitive recursive bound on the
% number of iterations of the backward reachability procedure.  However,
% the procedure has proved very useful in practice when solving the
% reachability problem for classes of systems (e.g., broadcast
% protocols) with the same worst case behavior.

\section{Decidability of related security analysis problems}
\label{app:related-problems}

Here we consider three security analysis problems which are related to
user-role reachability and discuss their decidability.

\paragraph{Inductive policy invariants.}
In~\cite{dynpal,passat09}, the problem of checking properties that
remain unaffected under any sequence of actions of arbitrary (but
finite) length is considered.  This is the dual problem of user-role
reachability; in fact, it is not difficult to prove that if the
backward reachability procedure terminates (with
$\mathsf{unreachable}$), then the fix-point is the strongest
invariant.  More precisely, In other words, a \emph{policy invariant}
is a formula which holds in every state of an ARBAC policy.  In our
framework, the problem of checking whether a property is an inductive
invariant (a particular case of a policy invariant) turns out to be
decidable because of Property~\ref{lem:dec-URA}.  Let
$\Gamma:=(In,\{\tau_i\},\{\iota_j\})$ be a symbolic ARBAC policy.  The
$\forall$-formula $\psi(ua)$ is an \emph{inductive (policy) invariant
  for $\Gamma$} iff (a) $In(ua)\Rightarrow \psi(ua)$ is valid modulo
$T_{ARBAC}$ and (b) $(\iota(ua)\wedge \psi(ua)\wedge
\tau(ua,ua'))\Rightarrow \psi(ua')$ is valid modulo $T_{ARBAC}$.  It
is easy to see that (a) and (b) can be reduced to the satisfiability
of BSR formulae.  In fact, (a) is equivalent to the unsatisfiability
modulo $T_{ARBAC}$ of $In(ua)\wedge \neg \psi(ua)$, which---in
turn---can be transformed to a formula of BSR.  Similarly, (b) is
equivalent to the unsatisfiability modulo $T_{ARBAC}$ of
$\iota(ua)\wedge \psi(ua)\wedge \tau(ua,ua')\wedge \neg \psi(ua')$
which is again logically equivalent to a BSR formula.  These
observations with Property~\ref{lem:dec-URA} imply the following fact.
\begin{theorem}
  The problem of establishing if a $\forall$-formula is an inductive
  policy invariant is decidable.
\end{theorem}
Indeed, checking inductive invariants is a lot cheaper than running
the backward reachability procedure.  The drawback is that if a
property $\psi$ fails to be an inductive invariant, then we cannot
conclude about its being an invariant of $\Gamma$ (in other words,
inductive invariants are a strict sub-class of policy invariants).
However, we can take the complement $\neg \psi$ (which is an
$\exists$-formula) of $\psi$ and run the backward reachability
procedure.  If this returns $\mathsf{unreachable}$, then we can
conclude that $\psi$ is an invariant of $\Gamma$.

\paragraph{Role containment.}
The problem of \emph{role containment} for a symbolic ARBAC policy
$\Gamma$ consists of checking if every member of a certain role, say
$e^r_1$, is also member of another role, say $e^r_2$, in every state
reachable from the initial state.  For simplicity, assume there is no
role hierarchy.  It is easy to reduce this to the user-role
reachability problem by considering a role $e^r_k$ not occurring in
$\Gamma$ and the following \textsl{can\_assign} action:
\begin{eqnarray*}
  \exists u,r,r_1. 
  \left(
  \begin{array}{l}
    ua(u,r) \wedge r=e^r_1 \wedge  
    \neg ua(u_1,r_1) \wedge r_1=e^r_2 \wedge 
     ua' = ua \oplus (u,e^r_k)
    \end{array}
    \right)  .
\end{eqnarray*}
Let $\Gamma'$ be obtained by adding the action above to $\Gamma$. It
is easy to see that the role containment problem for $\Gamma$ is
solvable iff role $e^r_k$ is reachable by $\Gamma'$.  

\paragraph{Weakest precondition.}
The \emph{weakest precondition problem} for a transition system
$\Gamma$ and goal $\gamma$ consists of computing the minimal sets of
initial role memberships of a given user $e^u_k$ for which $\gamma$ is
reachable.  This can be reduced to the user-role reachability problem
by taking $\forall u,r.\neg ua(u,r)$ as the initial state formula $In$
and then using a refinement of the backward reachability procedure in
Figure~\ref{fig:reach-algo}.  The refinement consists of using
$\exists$-formulae whose matrix is a conjunction of literals only;
this is without loss of generality as any $\exists$-formula can be
transformed to a finite disjunction of $\exists$-formulae whose
matrices are conjunctions of literals, called $\exists^+$-formulae, by
simple logical manipulations, and representing the search space by a
forest of trees whose nodes are labelled by $\exists^+$-formulae.

The root nodes are labelled by the $\exists^+$-formulae whose
disjunction is equivalent to the goal $\gamma$.  Then, we iteratively
extend each tree by selecting a node with no sons by adding as many
sons as the number of $\exists^+$-formulae which are equivalent to the
pre-image of the formula labelling the father node.  After the
creation of a node $n$, we check whether a fix-point has been reached
as follows.  First, we consider the formula $\psi$ labelling node $n$.
Second, we take the disjunction of the $\exists^+$-formulae labelling
all the nodes in the tree except $\psi$: it is not difficult to see
that this is equivalent to the content of the variable $B$ of the
procedure in Figure~\ref{fig:reach-algo}, i.e.\ it is the set of
backward reachable states.  Third, we check the satisfiability of
$\neg ((\iota \wedge \psi)\Rightarrow B)$, which is similar to the
check at line 2 in Figure~\ref{fig:reach-algo} except that $\psi$ is
an $\exists^+$-formula instead of an $\exists$-formula.  Because the
pre-order on configurations is a wqo, it is possible to show that this
procedure always terminates with finitely many trees.  At this point,
we collect all the $\exists^+$-formulae labelling the nodes of the
trees in the forest, compute the corresponding configurations (this is
always possible because of Lemma~\ref{lem:semantic-counterpart}, and
take only those sets where the interpretation of $ua$ has the minimal
number of occurrences of the user $e^u_k$ as the first component.
Since all the computation are effective, the procedure terminates.

By these reductions, we obtain the decidability of these two security
analysis problems.
\begin{theorem}
  The containment and weakest precondition problems are decidable.
\end{theorem}

\section{A worked-out example}
\label{app:br-illustration}

We consider a simple reachability problem in~\cite{stoller}.  There
are several simplifying assumptions made by the authors
of~\cite{stoller} that allow us to: (i) ignore permissions and focus
only on roles, (ii) the role hierarchy can be abstracted away, (iii)
there is just one administrative role and user capable of executing an
administrative action of assignment, and (iv) there exists just one
user to which administrative actions can be applied.  As a
consequence, a \textsl{can\_assign} action can be seen as pair $\langle C, r'
\rangle$ (where the administrative role has been omitted) while a `can
revoke action only identifies the role $r'$ to be revoked and ignore
the administrative role that is supposed to apply the action, hence
its specification will simply be $\langle r' \rangle$.  Under these
assumptions, the ARBAC policy considered in~\cite{stoller} consists of
the following \textsl{can\_assign} actions:
\begin{eqnarray*}
  \textsl{can\_assign}_1 & : & \langle \{e^r_1\}, e^r_2 \rangle, \\
  \textsl{can\_assign}_2 & : & \langle \{e^r_2\}, e^r_3 \rangle, \\
  \textsl{can\_assign}_3 & : & \langle \{e^r_2\}, e^r_3 \rangle, \\
  \textsl{can\_assign}_4 & : &   \langle \{e^r_3, \overline{e^r_4}\}, e^r_5 \rangle, \\
  \textsl{can\_assign}_5 & : & \langle \{e^r_5\}, e^r_6 \rangle, \\
  \textsl{can\_assign}_6 & : &   \langle \{\overline{e^r_2}\}, e^r_7 \rangle, \\
  \textsl{can\_assign}_7 & : & \langle \{e^r_7\}, e^r_8 \rangle, 
\end{eqnarray*}
where we have dropped the numerical subscript of the constant $e^u$
denoting a user because of assumption (iv); and the following `can
revoke actions:
\begin{eqnarray*}
  \textsl{can\_revoke}_1 & : &   \langle r_1 \rangle, \\
  \textsl{can\_revoke}_2 & : &     \langle r_2 \rangle, \\
  \textsl{can\_revoke}_3 & : &     \langle r_3 \rangle, \\
  \textsl{can\_revoke}_4 & : &     \langle r_5 \rangle, \\
  \textsl{can\_revoke}_5 & : &     \langle r_6 \rangle, \\
  \textsl{can\_revoke}_6 & : &     \langle r_7 \rangle .
\end{eqnarray*}
The initial state $s_0$ of the ARBAC system is the following:
\begin{eqnarray*}
  s_0(ua) := \{ (e^u,e^r_1), (e^u,e^r_4), (e^u,e^r_7) \} ,
\end{eqnarray*}
and the goal is to reach a state where the user $e^u$ can be assigned
to role $e^r_6$.  As said in~\cite{stoller}, the goal is not reachable
from the initial state.  Below, we explain how to show that this is
indeed the case in our framework and using the backward reachability
procedure in Figure~\ref{fig:reach-algo}.

First of all, we specify the theory $T_{\mathit{ARBAC}}:=
T_{\mathit{Role}}\cup T_{\mathit{User}}\cup T_{\mathit{Permission}}\cup PA$
as follows:
\begin{eqnarray*}
  T_{\mathit{Role}} & := & SV(\{e^r_1, ..., e^r_8\}, \mathit{Role}) \\
  T_{\mathit{User}} & := & SV(\{e^u\}, \mathit{User}) \\
  T_{\mathit{Permission}} & := & \emptyset \\
  PA & := & \emptyset . 
\end{eqnarray*}
The formula $In(ua)$ characterizing the set of initial states is
expressed by
\begin{eqnarray*}
  \forall u,r.(ua(u,r)  & \Leftrightarrow & 
  \left(
    \begin{array}{ll} 
      (u=e^u \wedge r=e^r_1) & \vee \\
      (u=e^u \wedge r=e^r_4) & \vee \\
      (u=e^u \wedge r=e^r_7) & 
    \end{array}
  \right).
\end{eqnarray*}
The goal formula $\gamma(ua)$ characterizing the set of goal states is
expressed by
\begin{eqnarray*}
  \exists u,r.(ua(u,r)  & \wedge & u=e^u \wedge r=e^r_6) .
\end{eqnarray*}
Notice that because of assumption (ii), the restricted form of
negation allowed in the preconditions of transitions of the form
(\ref{eq:can_assign}) is sufficient to precisely describe the `can
assign actions above:
\begin{eqnarray*}
  \textsl{can\_assign}_1 & : &
  \exists u,r.(ua(u,r)  \wedge r = e^r_1 \wedge 
  ua' = ua \oplus (u,e^r_2) \\
  \textsl{can\_assign}_2 & : &
  \exists u,r.(ua(u,r)  \wedge r = e^r_2 \wedge 
  ua' = ua \oplus (u,e^r_3) \\
  \textsl{can\_assign}_3 & : &
  \exists u,r,r_1.\left(
    \begin{array}{l}
      ua(u,r)  \wedge r = e^r_3 \wedge 
      \neg ua(u,r_1) \wedge r_1 = e^r_4 \wedge \\
      ua' = ua \oplus (u,e^r_5)
    \end{array}
  \right)       \\
  \textsl{can\_assign}_4 & : &
  \exists u,r.(ua(u,r)  \wedge r = e^r_5 \wedge 
  ua' = ua \oplus (u,e^r_6) \\
  \textsl{can\_assign}_5 & : &
  \exists u,r.(\neg ua(u,r)  \wedge r = e^r_2 \wedge 
  ua' = ua \oplus (u,e^r_7) \\
  \textsl{can\_assign}_6 & : &
  \exists u,r.(ua(u,r)  \wedge r = e^r_7 \wedge 
  ua' = ua \oplus (u,e^r_8) .
\end{eqnarray*}
The \textsl{can\_revoke} actions can be expressed as follows:
\begin{eqnarray*}
  \textsl{can\_revoke}_1 & : & \exists u.(ua' = ua \ominus (u,e^r_1)) \\
  \textsl{can\_revoke}_2 & : & \exists u.(ua' = ua \ominus (u,e^r_2)) \\
  \textsl{can\_revoke}_3 & : & \exists u.(ua' = ua \ominus (u,e^r_3)) \\
  \textsl{can\_revoke}_4 & : & \exists u.(ua' = ua \ominus (u,e^r_5)) \\
  \textsl{can\_revoke}_5 & : &\exists u.(ua' = ua \ominus (u,e^r_6)) \\
  \textsl{can\_revoke}_6 & : &\exists u.(ua' = ua \ominus (u,e^r_7)) .
\end{eqnarray*}

Now, we can explain how the backward reachability procedure works on
the example.  In order to simplify the presentation, in the following,
we use a variant of the backward reachability procedure in
Figure~\ref{fig:reach-algo}.  The differences are the following.
First, instead of considering the disjunction of all the possible
actions and compute the pre-images of the goal with respect to this
complex formula, we consider the pre-images of the goal with respect
each possible action separately.  Indeed, this allows us to write more
compact formulae and, since it is easy to see that pre-image
computation distributes over disjunction, it is sufficient to take the
disjunction of the pre-images computed with respect to a single action
to obtain the same formula computed by the procedure in
Figure~\ref{fig:reach-algo}.  Concerning the satisfiability checks,
while the reachability test can be done as soon as we obtain a
(satisfiable) pre-image with respect to a single action, the fix-point
check requires a bit of care.  In fact, after obtaining a
(satisfiable) pre-image, the fix-point is \emph{local} to that
pre-image in the sense that all the (satisfiable) pre-images with
respect to the remaining actions must also be checked for fix-point.
Hence, a \emph{global} fix-point is reached only when all the local
fix-point are successful.  Furthermore, each local fix-point check
must be done by conjoining the actual pre-image with conjunction of
the negation of each pre-image previously computed.  It is not
difficult to see that the global fix-point corresponds to the
fix-point check of the procedure in Figure~\ref{fig:reach-algo}.  

First of all, the backward procedures computes the pre-image of
$\gamma$ with respect to each \textsl{can\_assign} and \textsl{can\_revoke} actions.
To illustrate how one of the pre-image computation is done, let us
consider $Pre(\textsl{can\_assign}_4,\gamma)$, i.e.
\begin{eqnarray*} 
  \exists u_1,r_1.(ua'(u_1,r_1)  \wedge  u_1=e^u \wedge r_1=e^r_6) & \wedge \\
  \exists u_2,r_2.(ua(u_2,r_2)  \wedge u_2=e^u \wedge r_2 = e^r_5 \wedge 
  ua' = ua \oplus (u_2,e^r_6) 
\end{eqnarray*}
where variables have been renamed to disambiguate the scope of
applications of the existential quantifiers and $ua'$ is implicitly
existentially quantified.  The formula can be rewritten as follows:
\begin{eqnarray*}
  \exists u_1,r_1,u_2,r_2.
  \left(
    \begin{array}{l}
      ua'(u_1,r_1)  \wedge  u_1=e^u \wedge r_1=e^r_6) \wedge \\
      ua(u_2,r_2)  \wedge u_2=e^u \wedge r_2 = e^r_5 \wedge \\
      ua'= \lambda w,r.(\mathit{if}~(w=u_2\wedge r=e^r_6)~\mathit{then}~ \mathit{true}~ \mathit{else}~ua(w,r))
    \end{array}
  \right)
\end{eqnarray*}
by simple logical manipulations and recalling the definition of
$\oplus$.  Then, substituting the $\lambda$-expression we derive:
\begin{eqnarray*}
  \exists u_1,r_1,u_2,r_2.
  \left(
    \begin{array}{l}
      \lambda w,r.(\mathit{if}~(w=u_2\wedge r=e^r_6)~\mathit{then}~ \mathit{true}~ \mathit{else}~ua(w,r))(u_1,r_1)  \wedge \\
      u_1=e^u \wedge r_1=e^r_6) \wedge 
      ua(u_2,r_2)  \wedge u_2=e^u \wedge r_2 = e^r_5 \wedge \\
      ua'= \lambda w,r.(\mathit{if}~(w=u_2\wedge r=e^r_6)~\mathit{then}~ \mathit{true}~ \mathit{else}~ua(w,r))
    \end{array}
  \right)
\end{eqnarray*}
which can be furtherly simplified as follows by using
$\beta$-reduction:
\begin{eqnarray*}
  \exists u_1,r_1,u_2,r_2.
  \left(
    \begin{array}{l}
      (\mathit{if}~(u_1=u_2\wedge r_1=e^r_6)~\mathit{then}~ \mathit{true}~ \mathit{else}~ua(u_1,r_1))  \wedge \\
      u_1=e^u \wedge r_1=e^r_6) \wedge 
      ua(u_2,r_2)  \wedge u_2=e^u \wedge r_2 = e^r_5 \wedge \\
      ua'= \lambda w,r.(\mathit{if}~(w=u_2\wedge r=e^r_6)~\mathit{then}~ \mathit{true}~ \mathit{else}~ua(w,r))
    \end{array}
  \right) .
\end{eqnarray*}
Now, we observe that $u_1=u_2$ is valid modulo $T_{\mathit{ARBAC}}$
since $T_{\mathit{User}}$ constrains the set of users to be the
singleton set $\{ e^u\}$ and that $r_1=e^r_6$ holds because it occurs
in the formula above.  Hence, we can simplify the formula above as
follows:
\begin{eqnarray*}
  \exists u_1,r_1,u_2,r_2.(
  u_1=e^u \wedge r_1=e^r_6) \wedge 
  ua(u_2,r_2)  \wedge u_2=e^u \wedge r_2 = e^r_5 ) 
\end{eqnarray*}
where $ua'$ has been dropped since the equality $ua'= \lambda
w,r.(\cdots)$ is easily seen to be always satisfiable (this is so
because to make the equality true, it is sufficient to take $ua'$
equal to the $\lambda$-expression on the right) .  Finally, simple
considerations on the quantified variables allow us to simplify the
last formula even further so as to obtain:
\begin{eqnarray*}
  \exists u,r.(
  ua(u,r)  \wedge u=e^u \wedge r = e^r_5 ) , 
\end{eqnarray*}
whose matrix is a policy constraint, exactly as the matrix of
$\gamma$.  This is not an accident as it is possible to show that that
the class of existentially quantified formulae whose matrix is a
policy constraint are closed under pre-image computation.  Let $B_0$
be $\gamma$ and $B_1$ be the last formula above.  The backward
procedure performs a satisfiability check of the conjunction between
$In$ and $B_1$, i.e.\ of the following formula:
\begin{eqnarray*}
  \forall u,r.(ua(u,r)  & \Leftrightarrow & 
  \left(
    \begin{array}{ll} 
      (u=e^u \wedge r=e^r_1) & \vee \\
      (u=e^u \wedge r=e^r_4) & \vee \\
      (u=e^u \wedge r=e^r_7) & 
    \end{array}
  \right)) \wedge 
  \exists u,r.\left(
    \begin{array}{ll} 
      ua(u,r)  & \wedge \\
      u=e^u \wedge r = e^r_5 & 
    \end{array}
  \right) 
\end{eqnarray*}
so as to check whether the goal has been reached.  Skolemizing the two
existentially quantified variables, we obtain:
\begin{eqnarray*}
  \forall u,r.(ua(u,r)  & \Leftrightarrow & 
  \left(
    \begin{array}{ll} 
      (u=e^u \wedge r=e^r_1) & \vee \\
      (u=e^u \wedge r=e^r_4) & \vee \\
      (u=e^u \wedge r=e^r_7) & 
    \end{array}
  \right)) \wedge 
  \left(
    \begin{array}{ll} 
      ua(\tilde{u},\tilde{r})  & \wedge \\
      \tilde{u}=e^u \wedge \tilde{r} = e^r_5 & 
    \end{array}
  \right) ,
\end{eqnarray*}
where $\tilde{r}$ and $\tilde{u}$ are fresh constants.  Now, observe
that the universally quantified variable $u$ can only take one value
as we have assumed that the set of users contains just one element
$e^u$; hence it must be $\tilde{u}=e^u$.  So, we are left with the
problem of instantiating the universally quantified variable $r$.  The
decidability result of Property~\ref{lem:dec-URA} allows us to
consider only the instances of the formula where $u$ is instantiated
to $e^u$ and $r$ to $\tilde{r}$.  It is not difficult to see that the
resulting formula is unsatisfiable, thus entitling us to conclude that
the sets of states characterized by $B_1$ and $In$ are disjoint and
the goal state is not reachable by applying $\textsl{can\_assign}_4$.

Then, the backward procedure proceeds to check for a fix-point.  This
is equivalent to the validity of $B_1 \Rightarrow B_0$ or to the
unsatisfiability of its negation, namely $B_1\wedge \neg B_0$:
\begin{eqnarray*}
  \exists u,r.\left(
    \begin{array}{ll} 
      ua(u,r)  & \wedge \\
      u=e^u \wedge r = e^r_5 & 
    \end{array}
  \right) \wedge 
  \forall u,r.\neg (ua(u,r)   \wedge  u=e^u \wedge r=e^r_6) .
\end{eqnarray*}
As before, we Skolemize the existentially quantified variables so as
to obtain the following formula:
\begin{eqnarray*}
  \left(
    \begin{array}{ll} 
      ua(\tilde{u},\tilde{r})  & \wedge \\
      \tilde{u}=e^u \wedge \tilde{r} = e^r_5 & 
    \end{array}
  \right) \wedge 
  \forall u,r.\neg (ua(u,r)   \wedge  u=e^u \wedge r=e^r_6) .
\end{eqnarray*}
where $\tilde{u},\tilde{r}$ are fresh constants.  As before, because
of Property~\ref{lem:dec-URA}, without loss of generality, we can
restrict to consider the formula obtained by instantiating $u$ to
$e^u$ and $r$ to $\tilde{r}$: this time, however, we conclude that the
formula is satisfiable.  Thus, we have shown that a fix-point has not
been reached and we need to compute the pre-images of $B_1$ w.r.t.\
the all the \textsl{can\_assign} and \textsl{can\_revoke} actions.  However, before
computing the pre-images of $B_1$, we also need to compute the
pre-images of $B_0$ w.r.t.\ $\tau$ in $\{ \textsl{can\_assign}_i |
i=1,2,3,5,6\}\cup \{ \textsl{can\_revoke}_i | i = 1,.., 6\}$, i.e.\ for the
remaining assignments and revocations.  This turns out to be useless
as all the formulae obtained in this way characterizes sets of states
that are sub-sets of those specified by $\gamma$ or, in other words,
we have reached a (local) fix-point.  For the sake of conciseness, we
do not do this here.  However, the reader can verify this as a simple
exercise by following the steps taken above for computing
$Pre(\mathit{\textsl{can\_assign}}_4, U)$ and checking for safety and
fix-point.  Similar observations hold also for the pre-images of
$B_1$: it turns out that all these formulae implies $B_1$, i.e.\
several (local) fix-point have been reached, and are unsatisfiable
when considered in conjunction with $In$, i.e.\ they pass the safety
check.  As a consequence, we can conclude that we have reached a
(global) fix-point and the goal is not reachable.

\end{document}